\documentclass{article}
\usepackage[latin1]{inputenc}
\usepackage{enumerate,url}
\usepackage{amsfonts,amsmath,amssymb}
\usepackage{graphicx}

\usepackage{bm}
\usepackage[standard,thmmarks,thref]{ntheorem}
\usepackage[usenames,dvipsnames,svgnames,table]{xcolor}
\usepackage[normalem]{ulem}

\usepackage{arydshln} %Dashed line.  for dotted: \hdashline[0.5pt/5pt]
\usepackage[section]{algorithm}
\usepackage{algorithmic}

\usepackage{chicago}
\usepackage[colorlinks=true,citecolor=blue]{hyperref}
\providecommand{\citeyear}[1]{\cite{#1}}
\newcommand{\cOne}{\ensuremath{\mathcal{C}^1}}
\newcommand{\tuple}[1]{\langle{#1}\rangle}
\definecolor{mygray}{RGB}{150,150,150}
\newcommand{\Luka}{{\L}ukasiewicz}
\newcommand{\LIP}{\Luka{} Infinitely-valued Probabilistic}
\newcommand{\luka}{\ensuremath{\textrm{\L}_\infty}}
\renewcommand{\phi}{\varphi}
\renewcommand{\emptyset}{\varnothing}
\qedsymbol{\ensuremath{\blacksquare}}
\renewtheorem{example}{Example}[section]
\renewtheorem{lemma}[example]{Lemma}
\renewtheorem{definition}[example]{Definition}
\renewtheorem{proposition}[example]{Proposition}
\renewtheorem{theorem}[example]{Theorem}

\title{Quantitative Logic Reasoning}
\author{Marcelo Finger\\
		Department of Computer Science\\ 
		University of S\~ao Paulo, Brazil\\	 
		\url{mfinger@ime.usp.br}
}

\begin{document}

\maketitle

\begin{abstract}
	In this paper we show several similarities among logic systems that deal simultaneously with deductive and quantitative inference. We claim it is appropriate to call the tasks those systems perform as Quantitative Logic Reasoning.  Analogous properties hold throughout that class, for whose members there exists a set of linear algebraic techniques applicable in the study of satisfiability decision problems. In this presentation, we consider as Quantitative Logic Reasoning the tasks performed by propositional Probabilistic Logic; first-order logic with counting quantifiers over a fragment containing unary and limited binary predicates; and propositional \LIP{} Logic.
\end{abstract}

\section{Introduction}
\label{sec:intro}

Quantitative Logic Reasoning aims at providing a unified treatment to several tasks that involve both a deductive logic reasoning and some form of inference about quantities.  Typically, reasoning with quantities involves probabilities and/or cardinality assessments.  Superficially, we are dealing with such distinct quantitative inferential capabilities but it is our aim to clarify that, to some significant extent, these approaches share a considerable set of common features, which include, but are not restricted to:
\begin{itemize}
	\item similar reasoning tasks with quantities, which typically involve decision problems such as satisfiability or entailment assessments;
	\item similarly structured fragments that lead to the existence of normal forms;
	\item similar characterizations of consistency in terms of coherence;
	\item similar formulations based on Linear Algebra;
	\item similar decision algorithms employing SAT-based column generation;
	\item similar complexity of decision problems, which for the fragments covered in this work are ``only'' NP-complete.
\end{itemize}

We believe that the presence of such similarities elicits the grouping of several logic systems under the name of Quantitative Logic Reasoning systems.

We explore the shared properties of three logic systems with the aim of bringing forward the similarities as well as the particularities of each system.  For that, we present some well known results, which are employed as a basis for the development of quantitative reasoning techniques; we also present original results, mainly in dealing with counting quantifiers over unary and restricted binary predicates; and in the normal form and linear algebraic methods for \LIP{} Logic.  But the main claim of originality lies in bringing forward the similarities of all those systems.

The following logic systems are studied in detail.
\begin{itemize}
	\item Probabilistic Logic (PL).  It consists of classical propositional logic enhanced with probability assignments over formulas, presented in Section~\ref{sec:pl}.
	\item Counting Quantifiers over a first order fragment containing unary predicates; we show that such a fragment can be extended with binary predicates in restricted contexts without a complexity blow up.  The CQU and CQUEL logics are presented in Section~\ref{sec:cqu}.
	\item \LIP{} Logic (LIP), a multi-valued logic for which there exists a well-founded probability theory, presented in Section~\ref{sec:lukaAndProb}.
\end{itemize}

For each system above, we present language, semantics and decision problem, followed by  normal form presentation and satisfiability characterization We also present complexity results and decision algorithms.

It is important to note that throughout this work those logics and their decision problems are presented syntactically, and formulas are linguistic objects, presented as a context-free grammar or some similar, recursive, device. The syntactic vocabulary contains, at the level of terminals, a set of basic (propositional) symbols $\mathcal{P}$, a set of connectives with appropriate arity and punctuation symbols.

\section{Probabilistic Logic}
\label{sec:pl}

Probabilistic logic combines classical propositional inference with classical (discrete) probability theory. The original formulation of such a blend of logic and probability is due to George Boole who, in his seminal work introducing what is now known as Boolean Algebras, already dedicated the two last sections to the problem of combining logic and probability~results, stating that
\begin{quote}\it
	the object of the theory of probabilities might be thus defined. Given the probabilities of any events, of whatever kind, to find the probability 	of some other event connected with them.
	\smallskip 
	
	\hfill Boole ~\citeyear[Chapter XVI, 4, p.189]{Boo1854} 
\end{quote}

Deciding if a given set of probabilities is consistent or coherent may be seen as a first step for Boole's ``probability extension problem''. Indeed, there is certainly more than one way of computing probabilities starting from the establishment of their coherence; see~\cite{dFi2017} and also the methods presented in this work.

For the purposes of this work, we concentrate on the decision problem of probabilistic logic, the \emph{Probabilistic Satisfiability} problem (PSAT), which consists of an assignment of probabilities to a set of propositional formulas, and its solution consists of a decision on whether this assignment is satisfiable; this formulation is based on a full Boolean Algebra which, due to de Finetti's Dutch Book Theorem (see Proposition~\ref{pr:coherent} below), is equivalent to deciding the coherence criterion over a finite Boolean Algebra. The problem has been first proposed by Boole and has since been independently rediscovered several times (see~\cite{Hai86,HJ2000} for a historical account) until it was presented to the Computer Science and Artificial Intelligence community by Nilsson~\cite{Nil86} and was shown to be an NP-complete problem, even for cases where the corresponding classical satisfiability is known to be in PTIME~\cite{GKP1988}.

Boole's original formulation of the PSAT problem did not consider conditional probabilities, but extensions for them have been developed~\cite{Hai86,HJNA95,HJ2000,WPV04}; the latter two works also cover extensions of PSAT with imprecise probabilities. The complexity of the decision problems for conditional probabilities becomes PSPACE-complete if constraints can combine distinct conditional events; otherwise it remains NP-complete~\cite{FHM1990}. A few tractable fragments of PSAT were presented~\cite{AP2001}.  In this work, however, we concentrate on PSAT's original formulation, and in this section we follow the developments of \cite{FDB2011,DCF2014,DF2015,DF2015b}.

The PSAT problem is formulated in terms of a linear algebraic problem of exponential size.  The vast majority of algorithms for PSAT solving in the literature are based on linear programming techniques, such as column generation, enhanced by several types of heuristics~\cite{KP90,HJNA95,FDB2011,DF2015b}.

On the other hand, there is a distinct foundational approach to sets of probability assignment to formula known as \emph{coherent probabilities}, which are based on de Finetti's view of probabilities as betting odds~\cite{dFi1931,dFi37,dFi2017}.

In the following we present a few examples in Section~\ref{sec:expsat}, discuss the relationship between PSAT and coherent probabilities in Section~\ref{sec:coherentpsat} and present an algorithm for deciding PSAT in Section~\ref{sec:psatsolve}.
% and comment the topic of measuring the inconsistency of probabilistic bases in Section~\ref{sec:incmeas}.

\subsection{Examples}
\label{sec:expsat}

Consider the following example.
\begin{example} \label{ex:psat1}\rm
	A doctor is studying a disease $D$ and formulates a hypothesis, according to which there are three genes involved, $g_1, g_2$ and $g_3$ such that at least two of which must be present for the disease $D$ to manifest itself.  Studies in the population of $D$-patients shows that each of the three genes is present in 60\% of the patients.
	
	The question is whether the doctor's hypothesis is consistent with the data.\qed
\end{example}

In this example, we see a hypothesis consisting of hard statements (statements with probability 1) being confronted with probabilistic data.  The consistency of the joint statement is sees as decision problem of the sort we are dealing with here.

A second example is as follows.
\begin{example}\label{ex:ant}\rm
	In an ant colony infestation, three observers have reached different conclusions.
	\begin{itemize}
		\item Observer 1 noticed that at least 75\% of the ants had mandibles or could carry pieces of leaves.
		\item Observer 2 said that at most a third of the ants had mandibles or did not display the ability to carry pieces of leaves.
		\item Observer 3 stated that at most  15\% of the ants had mandibles. 
	\end{itemize}
	The question is whether these observations are jointly consistent or not. 
	\qed
\end{example}

We now see how these examples can be formalized.

\subsection{Coherent Probabilities and Probabilistic Satisfiability}
\label{sec:coherentpsat}

A \emph{PSAT instance} is a set $\Sigma=\{P(\alpha_i) \bowtie_i p_i| 1 \leq i \leq k\}$, where $\alpha_1, \ldots, \alpha_k$ are classical propositional formulas defined on $n$ logical variables\footnote{In computational logic tradition, variables are also called (syntactical) \emph{atoms}, but to avoid confusion with the algebraic use of `atom' as the smallest nonzero element of an algebra, we use here instead the term \emph{propositional symbol}, or \emph{(atomic) proposition}.} $\mathcal{P} = \{x_1, \ldots, x_n\}$, which are restricted by probability assignments
$P(\alpha_i) \bowtie_i p_i$, where $\bowtie_i\, \in \{=, \leq, \geq\}$ and $1 \leq i \leq k$. It is usually the case that all $\bowtie_i$ are equalities, in which case the PSAT instance can be seen simply as a set of pairs $\{(\alpha_i,p_i)|| i= 1, \ldots, k\}$.

There are $2^n$ possible propositional valuations $v$ over the logical
variables, $v: \mathcal{P} \rightarrow \{0,1\}$; each such valuation is truth-functionally extended\footnote{Thus, valuations can be seen as homomorphisms of the set of formulas into the two element Boolean Algebra $\{0,1\}$.}, as usual, to all formulas, $v: \mathcal{L} \rightarrow
\{0,1\}$, and a formula $\alpha$ is \emph{valid} if every valuation satisfies it, noted as $\models \alpha$. Let $V$ be the set of all propositional valuations. 

A \emph{probability distribution over propositional valuations%
\footnote{While the presentation here stays on the syntactical level, in algebraic terms this notion can be seen as a probability measure over the free boolean algebra, in the sense of~\cite{HT1948}.  Recall that a measure on $A$ is a function $\tau: A \to [0, 1]$ which is additive for incompatibles and also satisfies $\tau(1) = 1$. When $A$ is finite, as in the case here, every $a \in A$ equals the disjunction of the atoms it dominates, so $\tau$ is uniquely determined by its value at the set of (algebraic) atoms of $A$.  For every element $a \in A$ the value of $\tau(a)$ is the sum of the values $\tau(e)$ for all atoms $e \leq a$.}%
} $\pi: V \rightarrow [0,1]$, is a function that maps every valuation to a value in the real interval $[0,1]$ such that $\sum_{i=1}^{2^n} \pi(v_i) = 1$.  The probability distribution $\pi$ can be uniquely extended over the set of all propositional formulas built from $V$.  This, the probability of a formula $\alpha$ according to distribution $\pi$ is given by $P_\pi(\alpha) = \sum \{\pi(v_i)| v_i(\alpha) = 1\}$. The following is a straightforward consequence of this definition.

\begin{lemma} \label{lemma:kolmogorov}
	The probability $P_\pi$ defined above respects Kolmogorov's basic properties of discrete probability:
	\begin{quote}
	\begin{description}
		\item[K1] $0 \leq P_\pi(\alpha) \leq 1 $
		\item[K2] If $\models \alpha$ then $P_\pi(\alpha) = 1$
		\item[K3] If $\models \lnot(\alpha \land \beta)$ then $P_\pi(\alpha \lor \beta) = P_\pi(\alpha) + P_\pi(\beta)$
	\end{description}		
	\end{quote}
\end{lemma}

%\begin{proof}
%	\begin{description}
%		\item[K1] As $\pi(v_i) \geq 0$, $0 \leq P_\pi(\alpha) = \sum \{\pi(v_i)| v_i(\alpha) = 1\} \leq \sum_{i=1}^{2^n} \pi(v_i) = 1$
%		\item[K2] If $\models \alpha$, then $v(\alpha)=1$ for every $v$, so $P_\pi(\alpha)$.
%		\item[K3] AS it is never the case that $v(\alpha \land \beta) = 1, v(\alpha \lor \beta)=1$ exactly when $v(\alpha)=1$ or $v(\beta)=1$.
%	\end{description}probability
%\end{proof}

Nilsson~\citeyear{Nil86}'s linear algebraic formulation of PSAT considers a $k \times 2^n$ matrix $A = [a_{ij}]$ such that $a_{ij} = v_j(\alpha_i).$ The \emph{probabilistic satisfiability problem} is to decide if there is a probability vector $\pi$ of dimension $2^n$ that obeys the \emph{PSAT restriction}:
\begin{eqnarray}
	\nonumber
	A \pi &\bowtie& p\\
	\label{eq:PSATrestrictions}
	\mbox{$\sum \pi_i$} &=& 1\\
	\nonumber
	\pi &\geq& 0
\end{eqnarray}
where $\bowtie$ is a ``vector'' of comparison symbols, $\bowtie_i \in \{=, \leq, \geq\}$.

A \emph{PSAT instance} $\Sigma$ is \emph{satisfiable} iff its associated PSAT restriction~(\ref{eq:PSATrestrictions}) has a solution.  If $\pi$ is a solution to~(\ref{eq:PSATrestrictions}) we say that $\pi$ satisfies $\Sigma$.  The last two conditions of~(\ref{eq:PSATrestrictions}) force $\pi$ to be a probability distribution.  Usually the first two conditions of~(\ref{eq:PSATrestrictions}) are joined, $A$ is a $(k+1) \times 2^n$ matrix with 1's at its first line, $p_1=1$ in vector $p_{(k+1) \times 1}$, so $\bowtie_1$-relation is ``=''.

\begin{example}\label{ex:drunk1}\rm
  Consider Example~\ref{ex:psat1}. Let $x_i$ represent that gene $i$ is active in a $D$-patient.  The hypothesis that at least two genes are active in a given $D$-patient is represented by $\lnot(\lnot x_i \land \lnot x_j)$ with 100\% certainty for $i \neq j$:
  \[P(x_1 \lor x_2) = P(x_1 \lor x_3) = P(x_2 \lor x_3) = 1.\]
  The data stating that each gene occurs in 60\% of $D$-patients is given by:
  \[P(x_1) = P(x_2) = P(x_3) = 0.6,\]
  and the question is if there exists a probability distribution that
  simultaneously satisfies these 6 probability assignments.
  
  Consider now Example~\ref{ex:ant}. Let $x_1$ mean that an ant has mandibles and $x_2$ mean that that
  it can carry pieces of leaves.  In this case, we obtain the restrictions $\Sigma$:
  \[P(x_1 \lor x_2) \geq 0.75~~~~~
    P(x_1 \lor \lnot x_2) \leq 1/3~~~~~
    P(x_1) \leq 0.15\]
  
  Consider a probability distribution $\pi$ and all the possible
  valuations  as follows. 
  {\small
  	\[
  	\begin{array}{ccccccccc}
  	\pi & & x_1 & & x_2 & & x_1 \lor x_2 & & x_1 \lor \lnot x_2 \\
  	0.20 & & 0   & & 0   & & 0           & & 1 \\
  	0.05 & & 1   & & 0   & & 1           & & 1 \\
  	0.70 & & 0   & & 1   & & 1           & & 0 \\
  	0.05 & & 1   & & 1   & & 1           & & 1 \\
  	\cline{1-1} \cline{3-3} \cline{5-5} \cline{7-7} \cline{9-9}
  	1.00&& 0.10 && 0.75 && 0.80 && 0.30\\
  	\end{array}
  	\]
  }%
  which jointly satisfies the assignments above, so Example~\ref{ex:ant} is satisfiable. We are going to present an algorithm to compute one such probability distribution if one exists. \qed
\end{example}

On the other hand, de Finetti's approach aims at defining a ``coherent'' set of betting odds, or simply a coherent book.  Given a map from formulas to real values in $[0,1]$, $P:\{\alpha_1, \ldots,\alpha_k\} \to [0,1]$,  there is a \emph{Dutch book} against $P$ if there are $\sigma_1 , \ldots, \sigma_k \in \mathbb{R}$ such that

\[\sum_{i=1}^{k} \sigma_i ( P(\alpha_i) - v(\alpha_i) ) < 0 \textrm{ for all valuations $v$.}\]

The map is \emph{coherent} if there is no Dutch book against it.  

This can be understood as a game between  two players, Alice the bookmaker and Bob the bettor, wagging money on the occurrence of $\alpha_i$.  For each $i$, Alice states her betting odd $P(\alpha_i) = p_i \in [0,1]$ and Bob chooses a ``stake'' $\sigma_i \in \mathbb{R}$; Bob pays Alice $\sum_{i=1}^{k} \sigma_i \cdot P(\alpha_i)$ with the promise that Alice will pay back $\sum_{i=1}^{k} \sigma_i \cdot v(\alpha_i)$ if the outcome is possible world (or valuation\footnote{The notion of a ``world'', can be understood via Stone duality, whereby homomorphisms of a boolean algebra $A$ of events into the two element boolean	algebra $\{0,1\}$ are a dual counterpart of $A$, consisting of all possible evaluations of the events of $A$ into $\{0,1\}$, and can thus be identified with the set of possible worlds where these events take place.}) $v$.  The chosen stake $\sigma_i$ is allowed to be negative, in which case Alice pays Bob $|\sigma_i| \cdot P(\alpha_i)$ and gets back $|\sigma_i| \cdot v(\alpha_i)$ if the world turns out to be $v$. Alice's total balance in the bet is $\sum_{i=1}^{k} \sigma_i ( P(\alpha_i) - v(\alpha_i) )$.  So there is a Dutch book against Alice if the bettor has a choice of stakes such that, for every valuation $v$, Alice looses money.
Thus an assignment is coherent if for every set of stakes a bettor chooses, there is always a possible non-negative outcome.  It turns out that coherent maps are precisely those that can be seen as satisfiable PSAT instances.

\begin{proposition}[de Finetti~\citeyear{dFi1931,dFi37,dFi2017}]\label{pr:coherent}
	Given a  map from formulas to real values in $[0,1]$, $P:\{\alpha_1, \ldots,\alpha_k\} \to [0,1]$, the following are equivalent:
	\begin{enumerate}[(a)]
		\item $P$ is a coherent book.
		\item The probability assignment $\Sigma=\{(\alpha_i,P(\alpha_i))~|~i=1,\ldots,k\}$ is a satisfiable PSAT instance.
	\end{enumerate}
\end{proposition}

As a consequence of Proposition~\ref{pr:coherent} and Lemma~\ref{lemma:kolmogorov}, a coherent assignment is one that respects the axioms of probability theory.  Furthermore, to decide if an assignment is coherent, we can employ linear algebraic methods that solve \eqref{eq:PSATrestrictions}.

\begin{example}\label{ex:coherent}\rm
	In Example~\ref{ex:psat1}, consider a negative stake $\sigma= -1$ for the hypothesis information, and a positive stake of $\sigma=1$ for the probabilistic data, thus obtaining a total balance of
	\[S = -1\cdot((1-v(a \lor b))+(1-v(a \lor c))+(1-v(b \lor c))) + 1 \cdot((0.6-v(a))+(0.6-v(b))+(0.6-v(c)))\]
	It turns out that $S<0$ for all 8 possible worlds $v$, so this choice of stake constitutes a Dutch Book and the assignment is incoherent and, by Proposition~\ref{pr:coherent}, it is an unsatisfiable PSAT instance.\qed
\end{example}

\subsection{Algorithms for PSAT Solving}
\label{sec:psatsolve}

In this presentation, we follow Finger and De Bona~\citeyear{FDB2011,DF2015b}.

An important result of \cite{GKP1988}, which is an application of Carathéodory's Theorem~\cite{Eck93},  guarantees that a solvable PSAT instance has a ``small'' witness.

\begin{proposition}\label{fact:NP}
	If a PSAT instance $\Sigma=\{P(\alpha_i) = p_i| 1 \leq i \leq k\}$ 	is satisfiable, then there is a solution. $\pi$ to the PSAT restrictions~\eqref{eq:PSATrestrictions} such that there at most $k+1$ elements $\pi_j \geq 0$. \qed
\end{proposition}

Proposition~\ref{fact:NP} implies that the complexity of PSAT is in NP.  The special case where all $p_i = 1$ makes classical SAT a special case of PSAT, so PSAT is NP-hard. It follows that PSAT is NP-complete. 

A PSAT instance is in \emph{propositional normal form} if it can be partitioned in two sets, $\tuple{\Gamma,\Psi}$, where $\Gamma = \{P(\alpha_i) = 1 | 1 \leq i \leq m\}$ and $\Psi = \{P(y_i) = p_i | y_i \textrm{ is a propositional symbol, } 1 \leq i \leq k\}$, with $0 < p_i < 1$. The partition $\Gamma$ is the SAT part of the normal form, usually represented only as a set of propositional formulas and $\Psi$ is the \emph{propositional probability assignment} part.  By adding at most $k$  extra variables, any PSAT instance can be brought to normal form in polynomial time.

\begin{example}\label{ex:drunk2}\rm
	The PSAT instance in Example~\ref{ex:drunk1} is already in normal form, with $\Gamma = 	\{x_1 \lor x_2, x_1 \lor x_3, x_2 \lor x_3\}$ and $\Psi = \{P(x_1) = P(x_2) = P(x_3) = 0.6\}$.  This indicates that the normal form is a ``natural'' form in many cases, such as when one wants to confront a 	theory $\Gamma$ with the evidence $\Psi$.  %\qed

	For the formulation of Example~\ref{ex:ant}, we add three new variables, $y_1, y_2, y_3$ and make
	\[\Gamma 
		\begin{array}[t]{l}
		= \left\{
			\begin{array}{l}
				y_1 \to (x_1 \lor x_2) , ~
				(x_1 \lor \lnot x_2) \to y_2, ~
				x_1 \to y_3
			\end{array}
		  \right\} \\
		\equiv
		\left\{
			\begin{array}{l}
				x_1 \lor  x_2 \lor \lnot y_1,
				\lnot x_1 \lor y_2, x_2 \lor y_2,
				\lnot x_1 \lor y_3
			\end{array}
		\right\} 
		\end{array}
	\]
	and $\Psi = \{P(y_1) = 0.75, P(y_2) = \frac{1}{3}, P(y_3) = 0.15\}$. \qed
\end{example}

The algebraic formalization of PSAT~\eqref{eq:PSATrestrictions} has a special interpretation if the formula is in normal form, in which the columns of matrix $A$ are $\Gamma$-consistent valuations; a valuation $v$ over $y_1, \ldots, y_k$ is \emph{$\Gamma$-consistent} if there is an extension of $v$ over $y_1, \ldots, y_k, x_1, \ldots, x_n$ such that $v(\Gamma ) = 1$.  This property is the basis for encoding instances of PSAT into those of SAT.  However, due to the cubic increase in the number of variables, this method is too inefficient.  For details on this form of reduction, see~\cite{DF2015b}.

Instead, we plan to solve ~\eqref{eq:PSATrestrictions} without explicitly representing the exponentially large matrix $A$, using a method called~\emph{column generation}.  For that, we consider the following linear program:
 
\begin{align}\label{eq:program}
\begin{array}{ll}
	\min & c' \cdot \pi \\
	\mbox{subject to} 
		& A \cdot \pi=p \\
		& \pi\geq 0 \mathrm{~and~} \sum \pi_i = 1
	\end{array}
\end{align}

The cost vector $c$ in \eqref{eq:program} is a $\{0,1\}$-vector such that $c_i = 1$ iff column $A^j$ is $\Gamma$-inconsistent.  Thus, the column generation process proceeds by generating $\Gamma$-consistent columns.  The result of this minimization  process reaches total cost $c' \cdot \pi =0$ iff the input instance is satisfiable.  

We now describe the column generation process presented in Algorithm~\ref{algo:PSATviaSAT}, which solves \eqref{eq:program}. We start by describing the format of the input data. Condition $\sum \pi_i = 1$  in \eqref{eq:program} is usually incorporated in matrix $A$.  By convention, this equation always be the first line of $A$. Also by convention, vector $p$ is sorted in decreasing order, such that its first position contains a 1, corresponding to the equation $\sum \pi_i = 1$; accordingly, vector $p$ is prefixed with a 1.  Let $k=|\Psi|$. This convention allows us to solve the linear program \eqref{eq:program} initializing $A$ as an upper triangular matrix $T_\mathrm{up}$, which is a $(k+1) \times (k+1)$ square matrix where elements on the diagonal and above it are 1 and the remaining ones are 0.  As a consequence, the initial probability distribution $\pi$ is initialized such that $\pi_i = p_i - p_{i+1}$, $1 \leq i \leq k$ and $\pi_{k+1} = p_{k+1}$. The cost $c$ is a $\{0,1\}$-vector in which $c_j = 1$ iff column $A^j$ is $\Gamma$-inconsistent, $1 \leq j \leq k+1$.

In the column generation process, columns will be added to $A$, and the vectors for cost $c$ and solution $\pi$ will be correspondingly extended. As  all generated  columns at  the following  steps are  $\Gamma$-consistent, all cost elements added to $c$ are 0.

\begin{algorithm}[t] 
	\caption{\textit{PSATViaColGen}$(\phi)$ \label{algo:PSATviaSAT}}
	\textbf{Input:} a normal form PSAT formula $\tuple{\Gamma,\Psi}$.
	
	\textbf{Output:} a solution $(\pi,A)$ for \eqref{eq:program}, if one exists; ``No'', otherwise.
	
	\renewcommand{\algorithmiccomment}[1]{\hfill{// #1}}
	{\small
		\begin{algorithmic}[1]
			\STATE $A_{(0)} = T_\mathrm{up}$; compute cost vector $c^{(0)}$ and $\pi^{(0)}$ \label{lin:psatini}
			\FOR{$s=0;~c^{(s)}{'} \cdot \pi^{(s)} > 0;~s\!\!+\!\!+$} \label{lin:psatcgloop}
				\STATE $z^{(s)} = \textit{DualSolution}(A_{(s)}, p, c^{(s)})$ \label{lin:psatsimplex}
				\STATE $y^{(s)} = \mathit{GenerateColumn}(z^{(s)},\Gamma)$ \label{lin:psatcg}
				\RETURN ``No'' \textbf{ if } column generation failed 
				\STATE $A_{(s+1)} = \mathit{append-column}(A_{(s)}, y^{(s)})$ \label{lin:psatmerge}
				\STATE $c^{(s+1)} = \mathit{append}(c^{(s)},0)$
			\ENDFOR\label{lin:psatcgendloop}
			\RETURN $(\pi^{(s)},A_{(s)})$ such that $A_{(s)} \cdot \pi^{(s)} = p$ and $c^{(s)}{'}\cdot\pi=0$ \COMMENT{Successful termination}
		\end{algorithmic}}
\end{algorithm}

Column generation proceeds by steps. At step 0, we  start $A$, $c$ and $\pi$ as described above (line~\ref{lin:psatini}). At   each   step $s$,   we   start   by   solving    the   linear   program $A_{(s)}  \cdot \pi^{(s)}  =  p$ (line~\ref{lin:psatsimplex});  so we  suppose there is a  linear programming solver available; for an algorithm that does not presuppose a linear solver, see~\cite{FDB2011}. We require  that the solution generated  contains  the  \emph{primal  solution} $\pi^{(s)}$  as  well  as  the \emph{dual  solution} $z^{(s)}$~\cite{BT1997}; the dual solution is given by $z = c_B \cdot B^{-1}$, where $B$ is the basis of the linear program at step $s$, that is, a square sub-matrix of $A$ used to compute $\pi^{(s)}$ as the solution of $B \pi^{(s)}= p$, and $c_B$ is the cost of the columns of the basis.   These are  used  in the  column generation process  (line~\ref{lin:psatcg}) described below. If  column generation fails,    then    the   process    cannot    decrease    current   cost    and Algorithm~\ref{algo:PSATviaSAT} is terminated with  a negative decision.  Otherwise, a new column is generated and $A$ and  $c$ are expanded.  At the end, when the objective  function  has reached  0,  the  final values  of  $A$  and $x$  are returned. 

The idea of SAT-based column generation is to map a linear inequality over a set of $\{0,1\}$-variables into a SAT-formula, using the $O(n)$ method described in~\cite{War1998}. The inequality is provided by the column selection method used by the \emph{Simplex Method} for solving linear  programs~\cite{BT1997,PS1998}.  Given  a linear program  in format~\eqref{eq:program},  the
\emph{reduced  cost} $\bar{c}_y$  of  inserting column $y$  from  $A$ in a simplex basis is

\begin{align}\label{eq:psatcost}
	\bar{c}_y &= c_y - z' \cdot y 
\end{align}
where  $c_y$ is  the cost  associated  with column  $y$  and $z$  is the  dual solution of the  system $A \cdot \pi = p$ of size $k+1$.   As the generated column  $y$  is  always  $\Gamma$-satisfying, $c_y=0$,  so  to  ensure  a non-increasing value in the objective  function we need a non-positive reduced
cost, $\bar{c}_y \leq 0$, which leads us to% the inequality
\begin{align}\label{eq:psatredcost}
	z' \cdot y \geq 0
\end{align}
As $y$ is a $\{0,1\}$-vector, inequality~\eqref{eq:psatredcost} can be transformed into a SAT-formula; that formula is added to $\Gamma$ to obtain $\alpha$, which encodes a solution to~\eqref{eq:psatredcost} that is  $\Gamma$-satisfying.   We then  send $\alpha$ to a SAT-solver; if it is unsatisfiable, there is no way to  reduce the cost of the linear  program's objective function; otherwise, we obtain a satisfying assignment $v$. The generated column is $v(y)$, the restriction of $v$ the variables in $y$, which is a solution to~\eqref{eq:psatredcost}.  A new basis is obtained by substituting $v(y)$ for an appropriate outgoing column.  The Simplex Method provides a way of choosing the outgoing column, and guarantees the new basis is a feasible solution to linear program~\eqref{eq:program} whose cost is smaller than or equal to the previous one.

We have shown how to  construct a SAT-based column  generation function $\mathit{GenerateColumn}(z, \Gamma)$, provided  we   are given a (dual) solution for the corresponding linear program. 

\begin{theorem}
	Algorithms~\ref{algo:PSATviaSAT} and $\mathit{GenerateColumn}$ provide a decision procedure for the PSAT problem.
\end{theorem}

\begin{proof}  
	The correctness of Algorithm~\ref{algo:PSATviaSAT} is a direct consequence from the fact that $\tuple{\Gamma,\Psi}$ is satisfiable iff the linear program \eqref{eq:program} reaches a minimum at 0. As column generation only fails when it is impossible to decrease the cost function, this process either fails or brings the cost to 0, which is the only way Algorithm~\ref{algo:PSATviaSAT} terminates with a solution. \qed
\end{proof}

Note that the proof above guarantees termination, but even if it uses a polynomial-time linear solver, no polynomial bound is provided for the number of steps, which can in principle be $O(2^k)$.  Several implementations using the simplex method, employing various column generation strategies, are described in~\cite{DF2015b}, which also descrbe important empirical properties of those implementation.

\section{Counting Quantifiers over Unary Predicates}
\label{sec:cqu}

Counting quantifiers are quantitative constraints which may superficially look different from probabilistic reasoning.  However, here we demonstrate that there are striking similarities between these two forms of reasoning allowing us to lump them together under the heading of Quantitative Logic Reasoning.

The  need  to  combine  deductive  reasoning  with  counting  and  cardinality capabilities in  a principled way has  been long recognized, but  the  complexity of  this task has precluded its development. However, by generating a fragment of counting quantifiers inspired by the PSAT formulation, we are able to present a useful deductive system with counting that is ``only'' NP-complete and which allows for reasonably efficient, deterministic algorithms.

The basic approach for adding  counting capabilities extends first-order logic with some form of generalized quantifiers~\cite{Mos1957}, and we employ here a Lindstrom-type of quantifier~\cite{Lin1966} that can express the  counting  notions of  ``there  are  at least/most  $n$ elements  with  property $P$''.  Counting is first-order expressible, but it requires a first-order encoding using at least as many symbols as the counts one aims to express. On the other hand, the number of symbols employed by counting quantifiers is only proportional to the number of digits of the counts expressed. Hence expressing counting in first-order logic results in formulas whose size is exponentially larger than those obtained by using counting quantifiers.

The satisfiability of a logic  with  counting quantifiers, but limited  to a two-variable fragment  with at  most binary  predicates, is decidable~\cite{GKV1997,GO1999} with an EXPTIME-hard  lower  bound~\cite{Baa1996} and a NEXPTIME~\cite{Pra2005} upper-bound\footnote{Note that the fragment mentioned here has the finite model property~\cite{Pra2008}.}; recent studies on the complexity of specific counting problems are found in~\cite{MMS2015,BH2015}. Focusing on a one-variable fragment containing counting quantifiers over unary predicates only, the decision problem becomes NP-complete, even  when restricted only  to a  fragment  called Syllogistic  Logic, but the decision algorithm used to show that is inherently non-deterministic~\cite{Pra2008}.  

In a previous work, we presented an expressive fragment of first-order logic with counting  quantifiers over unary predicates called CQU~\cite{FDB2017}, which was developed applying techniques similar to those used in the PSAT case.  Here we extend the work on CQU by introducing CQUEL, for which the satisfiability problem remains NP-complete even as it partially allows the presence of binary predicates.  We start by presenting CQU, extend it to CQUEL,  and then develop decision algorithms for it, We start with a general example.

\begin{example}\label{ex:cquinit}\rm
	Consider the following group of people with several ages
	\begin{quote}
		\begin{enumerate}[(a)]
			\item At most 15 grandparents are married or happy. \label{it:hinit}
			\item At least 10 parents are not happy.
			\item At most 7 parents are not married.
			\item All grandparents are parents. \label{it:hend}
			\item At most 7 grandparents are unmarried and unhappy. \label{it:conc}
			\item At least 8 grandparents are unmarried and unhappy. \label{it:inc}
		\end{enumerate}
	\end{quote}
	We would like a way to determine that  (\ref{it:hinit})--(\ref{it:hend}) are satisfiable.  We would also like to have a method that allows us to infer (\ref{it:conc}) from those statements; equivalently, as (\ref{it:inc}) can be seen as the negation of (\ref{it:conc}), determine that (\ref{it:hinit})--(\ref{it:hend}) and (\ref{it:inc}) are jointly unsatisfiable.  These possibilities are all covered by the CQU formalism.  Moreover, suppose we are given a list of parent-children pairs (i.e.~a binary relation), and define  a parent as someone who has a child and, likewise, a grandparent as someone who has a grandchild.  To deal with this more general formulation, one needs a more expressive formalism such as CQUEL. \qed
\end{example}

In the following, we present the Semantics of CQU (Section~\ref{sec:cqusat}) and its extension CQUEL (Section~\ref{sec:cquel}). Then we present an algebraic formulation of the CQUEL-SAT problem (Section~\ref{sec:algebra}) which is used as a basis for the algorithms that solve it (Section~\ref{sec:algo}).

\subsection{Semantics and Satisfiability of CQU}
\label{sec:cqusat}

We now present formally a function-free one-variable first-order fragment over a signature containing only unary predicates and constants, extended with explicit counting quantifiers $\exists_{\leq n}$ (at most $n$) and $\exists_{\geq  n}$ (at  least $n$), where $n \in \mathbb{N}$ is  a non-negative integer. The semantics is tarskian, with models of arbitrarily large cardinality. 

The fragment contains two types of sentences over a countable set of variables $V$. Let $\psi(x)$ be a  Boolean combination of unary predicates $p(x),  q(x)$,   etc. A \emph{counting sentence} has the form $\exists_{\leq  n} x\, \psi(x)$ or  $\exists_{\geq n}  x\, \psi(x)$. A \emph{universal sentence} has the form $\forall x \psi(x)$.  A formula $\phi$ over the fragment of  counting quantifiers over unary predicates (CQU), is a  conjunction of any finite number of counting sentences $\mathcal{Q}$ and universal sentences $\mathcal{U}$, $\phi = \tuple{\mathcal{Q},\mathcal{U}}$.  Note that the universal and counting sentences involve only one-variable and only unary predicates; when we introduce the CQUEL fragment below a two-variable fragment will be involved, with restricted use of binary predicates\footnote{First-order one- and two-variable fragments are decidable, but the coding of counting quantifiers employs several new variables, so decidability is not immediate; see Proposition~\ref{prop:pratt}.}.

For the semantics, let the \emph{domain} $D$  be a non-empty set.  Let a \emph{term} be a constant or a variable.   Consider an \emph{interpretation}  $\mathcal{I}$; when applied to a term $t$, $\mathcal{I}(t) \in D$ and when applied to a  unary  predicate $p$, $\mathcal{I}(p) \subseteq  D$;  $\mathcal{I}_{|x}$ represents an interpretation that is identical to $\mathcal{I}$, except possibly for the interpretation of $x$. Let $\phi$ be a CQU-formula;  by $\mathcal{I} \models \phi$ we mean that  $\phi$ is satisfiable over $D$ with interpretation $\mathcal{I}$, defined as

\[
	\begin{array}{l@{~\mathit{iff}~}l}
		D,\mathcal{I} \models p(t) & \mathcal{I}(t) \in \mathcal{I}(p)\\
		D,\mathcal{I} \models \lnot \psi & D,\mathcal{I} \not\models \psi\\
		D,\mathcal{I} \models  \psi \land \rho &  D,\mathcal{I} \models \psi \textrm{~and~} D,\mathcal{I} \models \rho\\ 
		D,\mathcal{I} \models \exists_{\leq n} x\, \psi & \left|\{\mathcal{I}_{|x}(x) \in D  | D,\mathcal{I}_{|x}  \models \psi\}\right| \leq n\\
		D,\mathcal{I} \models \exists_{\geq n} x\, \psi & \left|\{\mathcal{I}_{|x}(x) \in D  | D,\mathcal{I}_{|x}  \models \psi\} \right| \geq n
	\end{array}
\]

The usual definitions apply to other Boolean connectives.  Note that the negation of counting sentences can be expressed within the CQU fragments, namely $\lnot \exists_{\leq n}  x\, \psi \equiv \exists_{\geq n+1}  x\, \psi$ and $\lnot \exists_{\geq n+1}  x\, \psi \equiv \exists_{\leq n}  x\, \pi$. The first-order  existential quantifier is expressed as $\exists x \psi \equiv \exists_{\geq 1} x\, \psi$ and  the universal quantifier as $\forall  x  \psi  \equiv  \exists_{\leq 0} x\, \lnot  \psi$.  The  \emph{exact counting quantifier} is defined as $\exists_{=n}   x\,  \psi   \equiv   \exists_{\leq  n}   x \, \psi   \land \exists_{\geq n} x\, \psi$.
 
If  there   are  $D$  and $\mathcal{I}$ such that $D, \mathcal{I} \models \phi$, then $\phi$ is a \emph{satisfiable} formula; otherwise it is unsatisfiable.     A    formula     $\phi$    \emph{entails} $\psi$ ($\phi   \models    \psi$)   iff  every pair $(D,\mathcal{I})$ that satisfies the former also satisfies the latter.    $\phi$   is \emph{equivalent} to   $\psi$
($\phi  \equiv  \psi$)  iff  they are satisfied by the same pairs $(D,\mathcal{I})$. The  problem CQU-SAT consists of deciding  whether a formula  is satisfiable.

If we remove the restriction to conjunctions of universal and counting sentences, we obtain the fragment called \cOne{}, studied in~\cite{Pra2008}.  Unlike CQU, \cOne{} allows for disjunctions between quantified formulas, such as $\exists_{\geq 7} x\, \psi \lor  \exists_{\leq 9} y\,  \rho$.  As the \cOne{} fragment has the finite model property and contains CQU, we obtain the following.

\begin{proposition}\label{prop:pratt}\textbf{\cite{Pra2008}} 
	Every satisfiable CQU formula is satisfiable over a finite domain. Moreover, CQU-SAT is strongly NP-complete. \qed
\end{proposition}

Strong NP-completeness means that when $n$ in $\exists_{\leq  n}, \exists_{\geq n}$  is given in unary notation, the decision remains NP-complete. As with probabilistic logic, we propose a normal form for formulas in the CQU fragment.  Existence of such normal form for CQU will allow us to extend the method to CQUEL.

\begin{definition} \label{def:nf}
	Let $\mathcal{U}$ be a finite set of universal sentences and let $\mathcal{Q}$ be a finite set of quantified unary predicates of the form	$\exists_{\leq n} x\, p(x)$ or $\exists_{\geq  n} x\, p(x)$, where $p$ is a unary atomic predicate. 	A \emph{normal form} CQU formula $\phi=\tuple{\mathcal{Q},\mathcal{U}}$ is the conjunction of formulas in $\mathcal{Q} \cup \mathcal{U}$.
\end{definition}

In the following we use $\bowtie$ to refer to $\leq$ or $\geq$, so  the  CQU  normal  form is  characterized  by  counting  quantifier sentences of the form $\exists_{\bowtie n} x\,p(x)$. By adding a small number of  extra predicates, any CQU  formula can be brought  to normal form.

\begin{lemma}\label{lemma:nf}
	For every CQU formula $\phi$ there exists a normal form formula $\phi'$ such that $\phi$ is a satisfiable iff $\phi'$ is; the normal form $\phi'$ can be built from $\phi$ in polynomial time.
\end{lemma}

\begin{proof}
	Consider       $\phi=\tuple{\mathcal{Q},\mathcal{U}}$.        We       build
	$\phi'=\tuple{\mathcal{Q}',\mathcal{U}'}$            starting           with
	$\mathcal{Q}'=\emptyset$  and $\mathcal{U}'=\mathcal{U}$.   Then, for  every
	quantified formula $\exists_{\bowtie  n} x\, \psi$, if $\psi$  is an unary
	predicate,  just  add $\exists_{\bowtie  n}  x\,  \psi$ to  $\mathcal{Q}'$;
	otherwise, create a new unary  predicate $p_\mathrm{new}$ and add $\forall x
	(p_\mathrm{new}(x)    \leftrightarrow   \psi)$    to   $\mathcal{U}'$    and
	$\exists_{\bowtie  n} x\,  p_\mathrm{new}(x)$ to  $\mathcal{Q}'$; at  every
	step $\phi'$ is in  normal form, and at its end,  by construction, $\phi'$ is
	satisfiable iff $\phi$ is. \qed
\end{proof}

\begin{example}\label{ex2}\rm
	Consider Example~\ref{ex:cquinit}, which can be formalized as follows:
	\begin{quote}
		\begin{enumerate}[(a)]\small
			\item $\exists_{\leq 15} x\, (g(x) \land (m(x) \lor h(x) ))$
			\item $\exists_{\geq 10} x\, (g(x) \land \lnot h(x))$ \label{ex2:unary}
			\item $\exists_{\leq  7} x\, (p(x) \land \lnot m(x))$
			\item $\forall x ( g(x) \to p(x) )$
			\item $\exists_{\leq 7} x\, ( g(x) \land \lnot m(x) \land \lnot h(x) )$ \label{ex2:pen}
			\item $\exists_{\geq 8} x\, ( g(x) \land \lnot m(x) \land \lnot h(x) )$ \label{ex2:last}
		\end{enumerate}
	\end{quote}		
	Clearly, (\ref{ex2:pen}) is the negation of (\ref{ex2:last}); we  use only the latter.  To bring counting formulas to normal form, introduce the predicates, $q_1, q_2, q_3, q_4$.  Let $\mathcal{U} = \{\forall x (q_1(x) \leftrightarrow (g(x) \land (m(x) \lor h(x) )), \forall x  (q_2(x) \leftrightarrow  (g(x) \land \lnot h(x)), \forall x (q_3(x) \leftrightarrow (p(x) \land \lnot m(x)), \forall x (g(x) \to  p(x)), \forall x (q_4(x) \leftrightarrow ( g(x) \land \lnot m(x) \land \lnot h(x) )\}$, so that we can have counting  quantification over unary predicates only; let $\mathcal{Q}=\{\exists_{\leq 15} x\,  q_1(x), \exists_{\geq 10} x\, q_2(x), \exists_{\leq 7} x\, q_3(x)\}$, such that we expect $\tuple{\mathcal{Q}, \mathcal{U}}$ to be satisfiable and $\tuple{\mathcal{Q} \cup \{\exists_{\geq 8} x\, q_4(x)\},  \mathcal{U}}$ to be unsatisfiable.\qed
\end{example}

It is important to note that the satisfiability problem for a set of CQU universal formulas is an NP-complete problem, for it can be reduced to a propositional problem.  

In fact, consider a normal form $\phi=\tuple{\mathcal{Q},\mathcal{U}}$.  Consider the $k = |\mathcal{Q}|$ unary predicates occurring in $\mathcal{Q}$, $p_1(x), \ldots, p_k(x)$;  then there  are $2^k$ \emph{elementary terms} of the form $e(x)  = \lambda_1(x)  \land  \ldots \land  \lambda_k(x)$, where each $\lambda_i(x)$  is either  $p_i(x)$ or  $\lnot p_i(x)$;  an elementary term $e(x)$  is called  \emph{susceptible} if it  is consistent  with the universal  sentences,  that  is,  the set  $\{\exists  x  e(x)\}  \cup \mathcal{U}$ has a model.

Semantically, each elementary term is interpreted as a \emph{domain elementary subset} $E \subseteq D$, $E=L_1 \cap \ldots \cap L_k$,  where each $L_i$ is  either the interpretation of $p_i$ or its complement  with respect to the domain $D$.  In any interpretation that satisfies $\phi=\tuple{\mathcal{Q},\mathcal{U}}$, only susceptible elementary terms  may be interpreted as  non-empty elementary subsets, otherwise the interpretation falsifies $\mathcal{U}$.

\begin{lemma}\label{lemma:npSusceptible}
	The problem of determining if there exists an elementary domain subset over the unary predicates in $\mathcal{Q}$ that is susceptible with a set of CQU universal formulas $\mathcal{U}$ is NP-complete.
\end{lemma}

\begin{proof}
	Transform  the  set $\mathcal{U}$  into   a  propositional  formula,  by   deleting  the  external $\forall  x$ quantifiers  and considering  each  unary predicate  $p(x)$ as  a propositional  symbol  $p$.   Then determining the existence of a satisfying valuation is an NP-complete problem.  If there is such a valuation, we obtain a susceptible element by considering a satisfying valuation $v$ restricted to the  proposition corresponding to the unary predicates in  $\mathcal{Q}$.  In that case, we consider a singleton domain $D=\{d\}$ and an interpretation $\mathcal{I}$ such that $d \in \mathcal{I}(p)$ iff $v(p)=1$.
\end{proof}

We now expand these results of~\cite{FDB2017} to include universal quantification over binary relations. The aim is to maintain the decision problem in NP.

\subsection{Expanding CQU into CQUEL}
\label{sec:cquel}

Previous results involving counting quantifiers and binary predicates brought the complexity of the satisfiability problem into EXPTIME~\cite{Baa1996}. However those methods allowed for counting quantification over sentences involving binary predicates.  The idea here is to maintain counting quantification over unary predicates in $\mathcal{Q}$, but to expand the set of allowed sentences in $\mathcal{U}$ so as to maintain the complexity of $\mathcal{U}$-satisfiability in NP.

Our idea is to expand $\mathcal{U}$ to allow for sentences corresponding to the first-order translation of statements from description logic DL Light~\cite{CDLLR2005}, thus lightly expanding CQU into CQUEL.  The satisfiability problem for DL Light is tractable, and we show here that adding those formulas to $\mathcal{U}$ leaves the satisfiability complexity in NP. There is at least another well known tractable description logic, $\mathcal{EL}^{++}$, which is however maximal with respect to tractability, in the sense that extending its language with the expressivity of CQU universal sentences brings the complexity to EXPTIME-complete~\cite{BBL05}.

The first-order signature now contains a finite set $\mathcal{P}$ of unary predicates, a finite set $\mathcal{R}$ of binary relations and a finite set $\mathcal{C}$ of constants.  The set of \emph{basic concepts} is the smallest set of unary expressions such that:
\begin{itemize}
	\item every unary predicate $p \in \mathcal{P} $ is a basic concept;
	\item if $r\in \mathcal{R}$, then $\exists y\, r(x,y) $ and $\exists y\, r(y,x) $ are basic concepts.
\end{itemize}

Basic concepts form \emph{concepts} in the following way.

\begin{itemize}
	\item every basic concept $B(x)$ is a concept;
	\item if $B(x)$ is a basic concept, then $\lnot B(x)$ is a concept;
	\item if $C_1(x)$ and $C_2(x)$ are concepts, so is $C_1(x) \land C_2(x)$.
\end{itemize}

A set $\mathcal{E}$ of \emph{extended light (EL) constraints} is a finite set of universal formulas of the following form.
\begin{quote}
	\begin{enumerate}[(a)]
		\item \emph{Inclusion Assertion (IA)}: $\forall x ( B(x) \to C(x) )$, where $B(x)$ is a basic concept and $C(x)$ is a concept; 
		\item \emph{Functionality Assertion (FA)}: $\mathsf{Funct}(r)$ and $\mathsf{Funct}(r^-)$, for $r \in \mathcal{R}$.  The semantics of $\mathsf{Funct}(r)$ states that if $(d,d'), (d,d'') \in \mathcal{I}(r)$, then $d' = d''$; similarly, the semantics of $\mathsf{Funct}(r^-)$ states that if $(d',d), (d'',d) \in \mathcal{I}(r)$, then $d' = d''$.
		\item \emph{Data:} $p(a), r(a,b)$ for $a,b \in \mathcal{C}$, $p \in \mathcal{P}$ and $r \in \mathcal{R}$.
	\end{enumerate}
\end{quote}

Note that EL constraints, except FAs, belong to a two-variable first-order fragment; FAs require the use of three variables, however in a very limited way.  It turns out that a the consistency of a set of EL constraints is not only decidable, but even tractable.  

A set of constraints $\mathcal{E}$ is \emph{negative inclusion (NI) closed} if for every IA $A=\forall x ( B_1(x) \to \lnot B_2(x) )$ above such that $\mathcal{E} \models A$, then $A \in \mathcal{E}$.  The \emph{NI-closure} of a set of EL constraints $\mathcal{E}$, $\overline{\mathcal{E}}$, is the smallest NI-closed set of constraints that contains $\mathcal{E}$.  The tractability of the satisfiability of a set of EL constraints follows from the following result.

\begin{proposition}[Calvanese \emph{et al.}~\citeyear{CDLLR2005}]\label{prop:eltractable}
	The NI-closure $\overline{\mathcal{E}}$ can be computed in polynomial time on the number of EL constraints in $\mathcal{E}$.
\end{proposition}

The proof of~Proposition~\ref{prop:eltractable} involves defining a normal form for constraints, then showing that computing the NI-closure of a set of constraints can be reduced in linear time to computing the NI-closure of a normalized set of constraints.  Then a set of constraint inference rules is proposed and it is shown that: (a) each rule application can be decided in polynomial time and the maximum number of rule applications is polynomial in $|\mathcal{E}|$; (b) each expansion rule inserts only an inferable constraint, from which soundness follows; (c) the inconsistency of $\overline{\mathcal{E}}$ can be obtained by a simple pattern search, which can be decided also in polynomial time. By composing all steps, we have a satisfiability check performed in polynomial time. Details in~\cite{CDLLR2005}.

We now define a CQUEL formula $\phi=\tuple{\mathcal{Q}, \mathcal{U}, \mathcal{E}}$ as a conjunction of the counting sentences in $\mathcal{Q}$, the universal CQU sentences in $\mathcal{U}$ and the EL constraints in $\mathcal{E}$. The semantics of a CQUEL formula $\phi$, $D, \mathcal{I}\models \phi$, is exactly as before, with the addition of binary predicates as in regular first-order logic.

\begin{example}\label{ex3}\rm
	Consider Example~\ref{ex2}, which we now extend with a binary relation $\mathit{parentOf}(x,y)$ representing the fact that $x$ is a parent of $y$.  Then we add the following set of EL constraints, stating that a parent is a parent-of someone.
	\[ \mathcal{E} = \left\{ \forall x \left( \frac{}{} p(x) \leftrightarrow \exists y\, \mathit{parentOf}(x,y)\right) \right\}\]
\end{example}

The previous result on the existence of normal forms applies to CQUEL too, in which counting quantification is applied only to unary atomic predicates.

\begin{lemma}\label{lemma:cquelnf}
	For every CQUEL formula $\phi$ there exists a polynomial-time computable normal-form formula $\phi'=\tuple{\mathcal{Q}, \mathcal{U}, \mathcal{E}}$ such that $\phi$ is a satisfiable iff $\phi'$ is, where $\mathcal{Q}$ contains only counting sentences over unary predicates.
\end{lemma}

\begin{proof}
	Following Lemma~\ref{lemma:nf}, the counting and universal sentences in $\phi$ are brought to normal form, the EL constraints of $\phi$ are also brought to normal form and added as $\mathcal{E}$ to $\phi'$.  Clearly, this can be done in polynomial time and satisfiability is preserved by Lemma~\ref{lemma:nf}.
\end{proof}

The following is a step into showing that the complexity of CQUEL is no greater than that of CQU.

\begin{lemma}\label{lemma:npU}
	The problem of deciding if a a set of formulas $\mathcal{U} \cup \mathcal{E}$ is consistent is NP-complete, where $\mathcal{U}$ is a set of CQU universal formulas and $\mathcal{E}$ is a set of EL constraints.
\end{lemma}

\begin{proof}
	Extend $\mathcal{E}$ into $\mathcal{E}'$, such that, for every existential constraint there is a new unary predicate equivalent to it; clearly this extension can be done in linear time and the input set $\mathcal{U} \cup \mathcal{E}$ is satisfiable iff $\mathcal{U} \cup \mathcal{E}'$ is. Compute the NI-closure $\overline{\mathcal{E}'}$ in polynomial time, by Proposition~\ref{prop:eltractable}.  The inconsistency of $\overline{\mathcal{E}'}$ can be decided in polynomial time, and if it is inconsistent so is the input set.

	If $\overline{\mathcal{E}'}$ is consistent, we construct $\mathcal{U}'$ by extending $\mathcal{U}$ with the IAs in $\overline{\mathcal{E}'}$ which contain only unary predicates; that is, there are no binary formulas in $\mathcal{U}'$. By Lemma~\ref{lemma:npSusceptible}, the coherence detection of $\mathcal{U}'$, and thus its satisfiability, is an NP-complete problem.
	
	Clearly, if $\mathcal{U}'$ is unsatisfiable, so is the input set.  If it is satisfiable, then we claim that $\mathcal{U} \cup \mathcal{E}$ is also satisfiable.  In fact, if $\mathcal{U}'$ is satisfiable, there is a model for it and by the proof of Lemma~\ref{lemma:npSusceptible} there is a model $(D,\mathcal{I})$ satisfying $\mathcal{U}'$. We extend this model in the following way.  Create a set $S$ of facts, initially empty.  For each fact $p(a)$ or $r(a,b)$ in $\overline{\mathcal{E}'}$, add this fact to $S$ and start the update propagation process.
	
	The update propagation process consists of the following.  If there is a constant $a$ in $S$ for which $\mathcal{I}(a)$ is not defined, create a new element and update $\mathcal{I}$ with  the constant and predicate interpretation; then propagate this update up the IA chain with constraints $\forall x\,(B(x) \to C(x)) \in \overline{\mathcal{E}'}$.  We have to deal with four cases. If $C(a) = p_i(a)$, add $p_i(a)$ to $S$ and propagate this change.  If $C(a) = \lnot p_i(a)$, do not update $\mathcal{I}$; due to the consistency of $\overline{\mathcal{E}'}$, we know that $p_i(a)$ will never be added to $S$.  If $C(a) = \exists y r_j(a,y)$ and there is no pair $(\mathcal{I}(a),d) \in I(r_j)$ skolemize $\exists y r_j(a,y)$ by taking a fresh constant $a'$, adding $r_j(a,a')$ to $S$, and propagate.  If $C(a) = \lnot \exists y r_j(a,y)$, do not update $\mathcal{I}$; again due to the consistency of $\overline{\mathcal{E}'}$, we know that $r_j(a,b)$ will never be added to $S$.  Similarly, no violation of a functionality assertion can occur, due to the consistency of $\overline{\mathcal{E}'}$.
	
	As the number of possible updates is finite, this update propagation process finishes in a finite number of steps, and we end up with an updated model $(D,\mathcal{I})$ that satisfies both $\overline{\mathcal{E}'}$ and $\mathcal{U}'$, and thus satisfies the input set, as desired.
\end{proof}

%We believe that CQUEL is a two-variable first logic with counting quantifiers over binary predicates, however limited, whose decision procedure is ``only'' NP-complete. 

The proof of Lemma~\ref{lemma:npU} gives us an Algorithm~\ref{algo:jointSAT} to determine the joint satisfiability of $\mathcal{U} \cup \mathcal{E}$. Line~\ref{line:SAT} employs a SAT-solver, such as~\cite{ES2003,Bie2014}.

\begin{algorithm}[t] 
	\caption{JointSAT$(\mathcal{U},\mathcal{E})$
		\label{algo:jointSAT}}
	\textbf{Input:} A set of CQU universal sentences $\mathcal{U}$ and a set of EL constraints $\mathcal{E}$.
	
	\textbf{Output:} If satisfiable, return a valuation representing a susceptible term; or ``No'', if unsatisfiable.
	
	{\small
		\begin{algorithmic}[1]
			\STATE Extend $\mathcal{E}$ into $\mathcal{E}'$, adding for every existential constraint a new unary predicate equivalent to it;
			\STATE Compute the NI-closure $\overline{\mathcal{E}'}$;
			\IF{$\overline{\mathcal{E}'}$ is inconsistent} 
				\STATE \textbf{return} ``No'';
			\ENDIF
			\STATE Extend $\mathcal{U}$ into $\mathcal{U}'$, adding the IAs in $\overline{\mathcal{E}'}$ which contain only unary predicates;
			\STATE Transform $\mathcal{U}'$ into a propositional formula $\alpha$, removing the quantifiers and variables;
			\STATE Apply SAT solver to $\alpha$; \label{line:SAT}
			\IF{$\alpha$ is satisfiable} 
				\STATE \textbf{return} satisfying valuation; 
			\ELSE 
				\STATE \textbf{return} ``No'' 
			\ENDIF
		\end{algorithmic}
	}
\end{algorithm}

The basic idea of Algorithm~\ref{algo:jointSAT} is to compose a formula to submit it to a SAT solver.  For that, the NI-closure of the input set of EL constraints is first computed.  If that already shows the problem is unsatisfiable, return.  Otherwise construct a propositional formula based on the ``propositional core'' of $\mathcal{U}$ and $\mathcal{E}$.  The final solution is obtained from applying a SAT-solver to this propositional formula.

For the rest of this work we always assume that formulas are in normal form. In the following, we look at CQUEL satisfiability in terms of integer linear algebra.  
%This will shed light on how to define  a susceptible view of counting assignments.

\subsection{Algebraic Formulation of CQUEL-SAT}
\label{sec:algebra}

Consider a normal form CQUEL formula $\phi=\tuple{\mathcal{Q},\mathcal{U}, \mathcal{E}}$ whose satisfiability we want to determine.  Consider the $k = |\mathcal{Q}|$ unary predicates occurring in $\mathcal{Q}$, $p_1(x), \ldots, p_k(x)$;  as in the CQU case, there  are $2^k$ \emph{elementary terms} of the form $e(x)  = \lambda_1(x)  \land  \ldots \land  \lambda_k(x)$, where each $\lambda_i(x)$  is either  $p_i(x)$ or  $\lnot p_i(x)$;  an elementary term $e(x)$  is called  \emph{susceptible} if it  is consistent  with $\mathcal{U} \cup \mathcal{E}$ ,  that  is,  the set  $\{\exists  x  e(x)\}  \cup \mathcal{U} \cup \mathcal{E}$ has a model.  Only susceptible elementary terms  may be interpreted as  non-empty elementary subsets, otherwise the interpretation falsifies $\mathcal{U} \cup \mathcal{E}$.

An  integer linear  algebraic presentation  of CQUEL-SAT is based on encoding each elementary term $e(x)$ as a $\{0,1\}$-vector $e$ of size $k$, in which $e_i = 1$ if  $\lambda_i$ is $p_i$ in $e(x)$ and $e_i=0$ otherwise.   We consider only the set of $k_m$ susceptible elementary terms, $0 \leq k_m \leq 2^k$.  Let $A$ be a $k \times k_m$ $\{0,1\}$-matrix, where each column encodes a susceptible elementary term; note that the $i$th line corresponds to the $i$th counting quantifier expression in $\mathcal{Q}$.  Let the $i$th element in $\mathcal{Q}$  be $\exists_{\bowtie_i n_i}  x\, p_i(x)$, $\bowtie_i \in\!  \{\leq,\geq\}$; let $b$ be a $k \times 1$ integer vector, such that $b_i = n_i$, and let $x$ be a $k_m \times 1$ vector of integer variables. Then  the  potentially  exponentially   large  integer  linear  system  that corresponds to the CQUEL-SAT problem $\phi=\tuple{\mathcal{Q},\mathcal{U}, \mathcal{E}}$ is:
\begin{eqnarray}
	\nonumber
	A x &\bowtie& b\\
	\label{eq:CQUELSAT}
	x & \geq & 0 \\
	\nonumber
	x_j &\makebox[0pt][l]{integer}
\end{eqnarray}

\begin{lemma}\label{lemma:cqusat}
	A normal form $\phi=\tuple{\mathcal{Q},\mathcal{U}, \mathcal{E}}$ is CQUEL satisfiable iff  its corresponding system given by~\eqref{eq:CQUELSAT} has a solution.
\end{lemma}

\begin{proof}
	$(\Rightarrow)$  If $\phi$  has  an  interpretation, let $x_j$ be the number of elements in the elementary subset corresponding to the $j$th susceptible elementary term; clearly $x_j$  is a  non-negative integer.   As  all  elements  in $\mathcal{Q}$  are  satisfied,  all inequalities in $Ax \bowtie b$ are satisfied.
	
	$(\Leftarrow)$ If  system~\eqref{eq:CQUELSAT} has  a solution, we construct a finite interpretation by inserting $x_j$ elements in each subset corresponding to a susceptible elementary term.  We can then compute an  interpretation  for   all  predicates  in  $\mathcal{Q}$,  and as all inequalities  in~\eqref{eq:CQUELSAT}  are   satisfied,  so  is  $\mathcal{Q}$; furthermore,  as  only susceptible  elementary  terms  have non-zero  elements, $\mathcal{U} \cup \mathcal{E}$ is also satisfied.
\end{proof}

To  determine  if  an  elementary  term is  susceptible, apply Algorithm~\ref{algo:jointSAT}, and consider the part of the returned valuation corresponding  to  the  predicates in  $\mathcal{Q}$,  as illustrated in the following example.

\begin{example}\label{ex4}\rm
	Consider the normal form formula $\phi=\tuple{\mathcal{Q},\mathcal{U}, \mathcal{E}}$ presented in Examples~\ref{ex2} and~\ref{ex3}. The linear algebraic rendering of the problem shows it is CQUEL satisfiable:
	\[
	\begin{array}{ll@{\hspace*{-1em}}l@{\hspace*{-1em}}l}
		\begin{array}{l}
			\begin{array}{l}
				\color{mygray}\mathbf{q_1}\\ 
				\color{mygray}\mathbf{q_2}\\ 
				\color{mygray}\mathbf{q_3}
			\end{array}\\\hdashline
			\begin{array}{l}
				\color{mygray}\mathbf{g}\\ 
				\color{mygray}\mathbf{p}\\ 
				\color{mygray}\mathbf{m}\\
				\color{mygray}\mathbf{h}
		\end{array}
	\end{array}
	&
	\begin{array}{l}
		\left[
		\begin{array}{ccc}
			0 & 1 & 0 \\
			1 & 0 & 0 \\
			0 & 0 & 1
		\end{array}
		\right] \cdot
		\\ \hdashline
		~\left.
		\begin{array}[t]{ccc}
			0 & 1 & 0 \\
			1 & 0 & 1 \\
			1 & 0 & 0 \\
			0 & 1 & 1			
		\end{array}
		\right.
	\end{array}
	&
	\begin{array}{l}
		\left[
		\begin{array}{c}
			10 \\ 0 \\ 0
		\end{array}
		\right]\\
		\begin{array}[t]{c}
		~ \\ ~ \\ ~ \\ ~
		\end{array}
	\end{array}
	&
	\begin{array}{l}
		\begin{array}{cc}
		\leq & 15\\
		\geq & 10\\
		\leq &  7
		\end{array}\\
		\begin{array}[t]{c}
			~ \\ ~ \\ ~ \\ ~
		\end{array}
	\end{array}
	\end{array}
	\]
	Then first three columns $\{0,1\}$-columns of size 7 are valuations over all unary predicates satisfying $\mathcal{U} \cup \mathcal{E}$; each valuation represents an elementary domain over predicates which are assigned 1 and the complement of the  0-assigned predicates. Each line  corresponds to a predicate, indicated on the left. The top three lines  contain the quantified restrictions in $\mathcal{Q}$ and  the matrix-vector product satisfies the  counting inequalities; the last four lines correspond to the predicates whose count are not quantified in $\mathcal{Q}$.  The three  .  This solution implies that the first four conditions of Examples~\ref{ex:cquinit} and \ref{ex2} are satisfiable.
	
	However, to show that adding the last condition leads to an unsatisfiable set of sentences, we would  have  to exhaustively consider the  $2^4$ valuations over predicates $q_1, \ldots, q_4$ and  show that that exponentially large system cannot satisfy the 4 inequalities.\qed
\end{example}

The exponential size of the proof search alluded by Example~\ref{ex4} can be avoided if there is a guarantee that all satisfiable CQUEL formulas have polynomial-sized models.  In the  case of  Probabilistic  Satisfiability (PSAT), which does not have the restriction on integral solution, the existence of polynomial-size models is guaranteed by Caratheodory's Theorem~\cite{Eck93}.  In the discrete case, we have the following analogue, which provides a polynomial-sized  bound  for  models of  satisfiable  CQUEL-SAT.

\begin{proposition}[Pratt-Hartmann~\citeyear{Pra2008}, E{\'e}n and S{\"o}rensson~\citeyear{ES2006}] \label{prop:size}  
	Consider a system of inequalities of the format~\eqref{eq:CQUELSAT} that has a positive integral solution.  Then it has a positive integral solution with at most $\left( \frac{5}{2} k\log k + 1\right) $ non-zero entries.
\end{proposition}

As presented in Algorithm~\ref{algo:jointSAT}, the satisfiability of $\mathcal{U}\cup\mathcal{E}$ can be represented by a $\{0,1\}$-valuation representing a susceptible term over its predicates.   Let $\{0,1\}$-matrix $A$ be  as in~\eqref{eq:CQUELSAT}; $A$'s $j$th column $A^j$ is  \emph{satisfying} if it represents the bits of a valuation returned by Algorithm~\ref{algo:jointSAT} on input $\mathcal{U}\cup\mathcal{E}$ restricted to $\mathcal{Q}$'s.  Lemma~\ref{lemma:cqusat} and Proposition~\ref{prop:size} yield the following.

\begin{lemma}  \label{lemma:size}  
	Consider  a normal  form  CQUEL-SAT  instance $\phi=\tuple{\mathcal{Q},\mathcal{U} \mathcal{E}}$. Then $\phi$ is satisfiable iff there exists a solvable system of inequalities of the form
	\begin{align}\label{eq:ineq}
	A_{k \times k_{m}} \cdot x_{k_{m} \times 1} & \bowtie ~b_{k \times 1} 
	\end{align}
	where $k_{m} \leq \left\lceil \frac{5}{2} (k\log k + 1) \right\rceil$, $A$ is a $\{0,1\}$-matrix whose columns satisfy $\mathcal{U}\cup\mathcal{E}$. \qed
\end{lemma}

This serves as a basis for effective algorithms for CQUEL-SAT.

\subsection{A CQUEL-SAT Solver based on Integer Linear Programming}
\label{sec:algo}

The  polynomial-size  format  of  solutions  given  by  Lemma~\ref{lemma:size} provides  a way  to reduce  a CQUEL-SAT to SAT; that is, an instance $\phi=\tuple{\mathcal{Q},\mathcal{U}}$ of  a CQUEL-SAT decision problem is  polynomially translated  to an  instance  of SAT  by encoding  the set  of inequalities in~\eqref{eq:ineq} such that the CQUEL-SAT is satisfiable iff its SAT translation is.  This approach is described in~\cite{FDB2017}, but the high number of variables in the translated SAT formulas, which is $O(k^3 \log k)$, makes this approach impractical in the critical areas of hard problems.  So a different approach, based on integer linear programming (ILP) and the branch-and-bound algorithm will be pursued.

The algebraic formulation of CQUEL-SAT on input $\phi=\tuple{\mathcal{Q},\mathcal{U}, \mathcal{E}}$ given by~\eqref{eq:CQUELSAT} is apparently  suited for Integer Linear Programming (ILP),  finding a solution  to $Ax  \bowtie  b$,  where  $x_j   \in  \mathbb{N}$.   However,  there  are  two important facts in~\eqref{eq:CQUELSAT} that have to be addressed, namely
\begin{itemize}
\item Matrix $A$ may be exponentially large.
\item As a consequence, we do not represent matrix $A$ explicitly; instead we deal with it partially and implicitly.
\end{itemize}

In fact, $A$'s columns consists of $\{0,1\}$-valuations representing susceptible terms satisfying $\mathcal{U}\cup\mathcal{E}$, which are  costly  to  compute  and  there  may  be  exponentially  many,  e.g.~when $\mathcal{U}=\mathcal{E}=\emptyset$.  To avoid these problems, we propose to solve the ILP problem via a simplified version of the  \emph{branch-and-bound} algorithm~\cite{Sch1986}, which  solves relaxed (continuous)  linear programs. As in the case of PSAT, we  generate $A$'s column as needed, in the process of \emph{column generation}~\cite{JHA1991} which takes place at each relaxed problem created by the branch-and-bound approach. For an ILP of the form~\eqref{eq:CQUELSAT}, it is  not necessary to search for  an optimal integer solution, one  only needs to find  a feasible one or show none exists.

\begin{algorithm}[t] 
  \caption{CQUELBranchAndBound$(\phi)$ \label{algo:CQUELBNB}}
  \textbf{Input:} A normal form CQUEL formula $\phi=\tuple{\mathcal{Q},\mathcal{U}, \mathcal{E}}$.

  \textbf{Output:} A solution satisfying~\eqref{eq:ineq}; or ``No'',
  if unsatisfiable.

{\small
  \begin{algorithmic}[1]
    \STATE $\mathit{CQUELSet} = \{\phi\}$
    \STATE $\mathit{SAT} = \mathrm{false}$
    \WHILE{not \textit{SAT} and $\mathit{CQUELSet}$ is not empty}
    \label{lin:loop}
        \STATE $\mathit{CQUELProblem} = \mathit{RemoveHeuristically}(\mathit{CQUELSet})$ \label{lin:heur1}
        \STATE $\mathit{solution} = \mathit{SolveRelaxedViaColGen}(\mathit{CQUELProblem})$  \label{lin:relax}
	    \IF{no \textit{solution} found} 
	      \STATE \textbf{continue} 
	    \ELSIF{integral \textit{solution}} 
	      \STATE $\mathit{SAT} = \mathrm{true}$
	    \ELSE 
	      \STATE $\textit{var} = \textit{choseBranchVar}(\mathit{solution})$ \label{lin:heur2}
	      \STATE $\mathit{newCQUELs} = \mathit{boundedProblems}(\mathit{CQUELProblem},var)$\label{lin:newprobs}
	      \STATE $\mathit{CQUELSet} = \mathit{CQUELSet} \cup \mathit{newCQUELs}$ 
	    \ENDIF
    \ENDWHILE\label{lin:endloop}
    \IF{ SAT }
      \RETURN \textit{solution}
    \ELSE
      \RETURN ``No''
    \ENDIF
  \end{algorithmic}
}
\end{algorithm}

The branch-and-bound method traverses an implicit search tree of relaxed problems. The top level of this search method  is   shown  in Algorithm~\ref{algo:CQUELBNB}. It starts in the root of the search tree with a unary set of problems containing the input CQUEL  formula, and it loops until either  a feasible integer solution to  the corresponding  linear  algebraic problem  given by~\eqref{eq:ineq}  is found or the set of problems becomes empty, in which case an unsatisfiability decision is reached. In the   main   loop    (lines   \ref{lin:loop}--\ref{lin:endloop}),  a problem is heuristically selected from   the    set   of    problems (line~\ref{lin:heur1}),  and its  \emph{relaxed version} is solved, which consists of the same problem without the restriction of integral solutions.   The heuristics implemented orders the problems according to the relaxed solutions to its parent in the tree, giving preference to solutions with the largest number of integer components.

If the relaxed problem  has no solution, it is  just removed from the set, which corresponds to closing a branch in the search tree, and the next iteration starts searching at an open branch.  If there is an integer solutio, the problem  is satisfiable and the  loop ends. Otherwise, a  solution with at least one non-integral element exists.  A  second heuristics is used to find a variable $x_{i}$ with  a non-integral solution $z_{i}$ on  which to branch (line~\ref{lin:heur2}), creating two new branches on the search tree.   This  heuristics  chooses  $x_{i^*}$  for  which  the non-integral $z_{i^*}$ is  closer to either $\left\lfloor z_{i^*} \right\rfloor$ or $\left\lceil z_{i^*} \right\rceil$.

The branching generates two  new bounded  problems $\phi'=\tuple{\mathcal{Q}',\mathcal{U}', \mathcal{E}}$, $\phi''=\tuple{\mathcal{Q}'',       \mathcal{U}'', \mathcal{E}}$ (line~\ref{lin:newprobs}),  with  the  creation   of  a  new  unary  predicate $p_{\mathrm{new}}$.  Note that the set of constraints $\mathcal{E}$ is never changed.  We make $\mathcal{U}'  = \mathcal{U}'' = \mathcal{U} \cup \{\forall   x  (p_{\mathrm{new}}(x)   \leftrightarrow  e_{i^*}(x))\}$,   where $e_{i^*}(x)$  is  the  elementary  term  corresponding  to  column  $i^*$  and $\mathcal{Q}'  =  \mathcal{Q}  \cup   \{\exists_{  \leq  \left\lfloor  z_{i^*}   \right\rfloor}  x\,p_{\mathrm{new}}(x)\}$ and  $\mathcal{Q}'' =  \mathcal{Q} \cup       \{\exists_{\geq       \left\lceil       z_{i^*} \right\rceil} x\,p_{\mathrm{new}}(x)\}$.  These  new formulas $\phi'$ and  $\phi''$ are then dealt with as integer linear problems of larger size. However,  if   their  size exceeds  the limit given  by Lemma~\ref{lemma:size}, the problem is not inserted.  

The  largest   part  of  the  processing   in \textit{CQUELBranchAndBound} occurs during the calls to the  relaxed solver (line~\ref{lin:relax}), \textit{SolveRelaxedViaColGen}$(\phi)$, in which column generation takes place. This process is analogous to that used for PSAT column generation, and it takes as input a CQUEL formula, eventually expanded by the bounding operation  and is  described  in Algorithm~\ref{algo:Relax}.   Its output  may contain some non-integral values, but  the objective function, which minimizes the solution cost has to be  0 for success to be  achieved.  Thus \textit{SolveRelaxedViaColGen}$(\phi)$  aims at  solving the  following linear program~\cite{BT1997}:

\begin{align}\label{eq:lp}
  \begin{tabular}{rl}
    minimize & $c' \cdot x$\\
    subject to & $A \cdot x \bowtie b$ and $x \geq 0$    
  \end{tabular}
\end{align}

In the linear program~\eqref{eq:lp}, $\{0,1\}$-matrix $A$'s columns consist of all possible  valuations over $k = |\mathcal{Q}|$ predicates and it  has $2^k$ columns.  The cost vector  $c$ and solution  vector  $x$  also  have  size  $2^k$, so  neither  is  represented explicitly.  Instead,  Algorithm~\ref{algo:Relax} starts with a  square matrix and iterates by generating the columns of $A$ in such a way as to decrease the objective function (lines~\ref{lin:cgloop}--\ref{lin:cgendloop}).

\begin{algorithm}[t] 
  \caption{\textit{SolveRelaxedViaColGen}$(\phi)$ \label{algo:Relax}}
  \textbf{Input:} A normal form CQUEL formula $\phi=\tuple{\mathcal{Q},\mathcal{U}, \mathcal{E}}$.

  \textbf{Output:} A relaxed solution $(A, x)$ , if it exists; or ``No'', if unsatisfiable.
  
  \renewcommand{\algorithmiccomment}[1]{\hfill{// #1}}
{\small
  \begin{algorithmic}[1]
    \STATE $A_{(0)} = I$; compute cost vector $c^{(0)}$; $x^{(0)} = b$ \label{lin:ini}
    \FOR{$s=0;~c^{(s)}{'} \cdot x^{(s)} > 0;~s\!\!+\!\!+$} \label{lin:cgloop}
      \STATE $z^{(s)} = \textit{DualSolution}(A_{(s)}, \bowtie b, c^{(s)})$ \label{lin:simplex}
      \STATE $y^{(s)} = \mathit{CQUELGenerateColumn}(z,\mathcal{U}, \mathcal{E})$ \label{lin:cg}
      \RETURN ``No'' \textbf{ if } column generation failed 
      \STATE $A_{(s+1)} = \mathit{append-column}(A_{(s)}, y^{(s)})$ \label{lin:merge}
      \STATE $c^{(s+1)} = \mathit{append}(c^{(s)},0)$ \label{lin:cost}
    \ENDFOR\label{lin:cgendloop}
    \RETURN $A_{(s)}$, $x^{(s)}$ such that $A_{(s)} x^{(s)} \bowtie b$ \COMMENT{Successful termination}
  \end{algorithmic}}
\end{algorithm}

As Algorithm~\ref{algo:Relax} is very similar to the column generation process for PSAT presented by Algorithm~\ref{algo:PSATviaSAT}, we only discuss here the main differences between the two.

As we do not have a restriction to ``add to one'' of PSAT, the initial size of $A$ is $k \times k$, and similarly the cost function $c$ starts with size $k$ and the bound vector $b$ has size $k=|\mathcal{Q}|$. As for the initialization (line~\ref{lin:ini}), $A$ receives the identity matrix $I$, and the solution $x$ receives $b$. The initialization of the $\{0,1\}$-cost vector, like in PSAT, is such that $c_j = 1$ iff column $A^j$ is $(\mathcal{U} \cup \mathcal{E})$-unsatisfiable.  The added columns will always be $(\mathcal{U} \cup \mathcal{E})$-satisfiable and receive cost 0 (line~\ref{lin:cost}).

As for the similarities, the steps within the loop are exactly the same for both algorithms, and for the same reason.  The goal of those steps is to decrease the cost function until it becomes 0, or fail if this is not possible.  

The only important difference in the loop is the column generation method.  Like in the PSAT case, it uses the dual solution $z$ to compute an inequality based on the reduced cost:
\begin{align}\label{eq:redcost}
	z' \cdot y \geq 0
\end{align}
Then it encodes the inequality \eqref{eq:redcost} to a propositional formula, which can be seen as a universal formula over unary predicates $\mathcal{U}'$.  It then calls Algorithm~\ref{algo:jointSAT} in the form $\mathit{JointSAT}(\mathcal{U} \cup \mathcal{U}',\mathcal{E})$ and if it is satisfiable, returns a valuation for its unary predicates.

\begin{theorem}
  Algorithms~\ref{algo:CQUELBNB},~\ref{algo:Relax} and $\mathit{GenerateColumn}$
  provide a decision procedure for the CQUEL-SAT problem.
\end{theorem}

\begin{proof}  \textsc{(Sketch)}   The  proof  is  a   simplification  of  the
  correctness of  the branch-and-bound  method for ILP~\cite{Sch1986},  due to
  the  fact that  CQUEL-SAT requires  only  a single  feasible integer  solution
  instead  of searching  for optimality  in  the lattice  of feasible  integer
  solutions.  Details omitted.
\end{proof}

There is an open source implementation\footnote{Available at \url{http://cqu.sourceforge.net} .} for CQU, that is CQUEL with $\mathcal{E} = \emptyset$.  It was developed in C++ and employs an open source linear programming solver\footnote{\url{http://www.coin-or.org/}} and the MiniSAT solver\footnote{\url{http://minisat.se/}} as part of the column generation process.  More details can be found at~\cite{FDB2017}.

\subsection{Future Challenges for Counting Quantifiers}
\label{sec:coherentcount}

We have shown that similar methods can be applied both for Probabilistic Logic and for Counting Quantifiers over unary predicates.  Three immediate challenges are suggested by this work.

The first one, which is a direct application of the expansion from CQU to CQUEL, is the application of the counting quantifier techniques developed here to the domain of Description Logics.  In particular, it would be nice to have an implementation for CQUEL and its deployment together with the existing tools for Description Logic Reasoning.

The second challenge is more foundational and comes directly from a comparison between results for Probabilistic Logic and Counting Quantifiers, namely, the search for a de Finetti-like notion of coherence for counting quantifiers.  In other words, this research topic searches for a betting foundation on counting quantifier statements in analogy to the Probabilistic Logic results described in Section~\ref{sec:coherentpsat}.

% See discussion deleted at version 2815.

The third challenge also comes by analogy with Probabilistic Logic, and it has to do with the existence of inconsistency measures for logic bases involving counting quantifiers.  This future investigation may take into consideration that it is possible that the analogy between probabilities and discrete counting breaks at this level, for the simple reason that inconsistency measures for probabilistic bases are continuous and may be approached by convex optimization methods~\cite{DF2015}, while counting quantifier treatment is discrete and based on integer linear programming techniques, are non-convex.

\section{\Luka\ Infinitely-valued Logic and Probabilities}
\label{sec:lukaAndProb}
	
\Luka\ Infinitely-valued logic is arguably one of the best studied many-valued logics~\cite{CDM2000}.  It has several interesting properties; semantically, formulas can be seen as taking values in the interval $[0,1]$; the semantics is truth functional, so then truth value of compound formula is function of the truth values of its components, and that function is continuous over the interval $[0,1]$; in fact, it is piecewise linear.  When truth values of propositional symbols are restricted to $\{0,1\}$, the semantics of formulas is that of classical logic; furthermore, it possesses a well developed proof-theory and an algebraic semantics base on MV-algebras. 

We present the essentials of Lukasiewicz (always propositional) logic ($\textrm{\L}_\infty$) and its underlying probability theory. We then introduce the notion of LIP-coherence, which is inspired on de Finetti?s notion of a coherent betting system. We defineand solve the LIP-satisfiability problem mimicking our analysis of the PSAT and
CQUEL-SAT problems.

\subsection{\Luka\ Infinitely-valued Logic}
\label{sec:luka}

Consider a finite set of propositional symbols $\mathcal{P} = \{p_1, \ldots, p_n\}$. We employ $\odot$ and $\oplus$ for \Luka\ conjunction and disjunction and write $\lnot$ for negation. Usually, only $\lnot$ and $\oplus$ are considered basic connectives. So all propositional symbols are formulas and if $\alpha$ and $\beta$ are formulas in  $\luka$, so are $\lnot \alpha$ and $\alpha \oplus \beta$.  Define $\alpha \odot \beta$  as $\lnot(\lnot\alpha \oplus \lnot\beta)$ and \Luka\ implication $\alpha \to \beta$ as $\lnot\alpha \oplus \beta$; it is also,possible to express the lattice connectives  $\alpha \land \beta$ as $\lnot(\alpha \oplus \lnot \beta) \oplus \alpha$ and $\alpha \lor \beta$ as $\lnot(\lnot\alpha \land \lnot\beta)$.

The semantics of $\luka$-formulas is given in terms of the rational (or real) interval $[0,1]$.  A \emph{valuation} is a map $v: \mathcal{P} \to [0,1]$ which is truth functionally extended to all \luka-formulas in the following way:

\[
\begin{array}{r@{~=~}l}
	v(\lnot \alpha) & 1 - v(\alpha )\\
	v(\alpha \oplus \beta) & \min( 1, v(\alpha) + v(\beta )) \\
	v(\alpha \odot \beta) & \max( 0, v(\alpha) + v(\beta ) - 1)
\end{array}
\]

The third line above can, of course, be obtained from the definition of $\odot$ in terms of $\lnot$ and $\oplus$. Similar truth functional expressions can be obtained for the other connectives:
\begin{align*}
	v(\alpha \to \beta ) =& min(1, 1- v(\alpha)+v(\beta))\\
	v(\alpha \land \beta) =& \min(v(\alpha),v(\beta))\\
	v(\alpha \lor \beta) =& \max(v(\alpha),v(\beta))
\end{align*}  

A formula $\alpha$ is \emph{valid} if $v(\alpha) = 1$ for every valuation $v$,  a formula $\alpha$ is \emph{satisfiable} (sometimes called 1-satsfiable) if there exists a $v$ such that $v(\alpha) = 1$; otherwise it is \emph{unsatisfiable}. A set of formulas $\Gamma$ is satisfiable if there exists a $v$ such that $v(\gamma) = 1$ for all $\gamma \in \Gamma$.  If $v(\alpha) = 1$, we say that $\alpha$ is satisfied by $v$.

It mis easy to see that $\alpha \to \beta$ is satisfied by $v$ iff $v(a) \leq v(b)$. If we define  $\alpha \leftrightarrow \beta$ as an abbreviation for $(\alpha \to \beta) \land (\beta \to \alpha)$, it follows that $\alpha \leftrightarrow \beta$ is satisfied by $v$ iff $v(\alpha) = v(\beta)$.

\subsection{\LIP\ Logic and \texorpdfstring{\luka}{L}-Coherence}
\label{sec:LIP}

\luka-valuations over propositional symbols $\{p_1, \ldots, p_n\}$ can be seen as points in and $n$-cube $[0,1]^n$. To apply the ideas and methods of Quantitative Logic Reasoning to probabilistic \luka, we follow the approach and terminology of~\cite{Mun2011}. Define a \emph{convex combination} of a finite set of valuations $v_1, \cdots, v_m$ as a function on formulas into $[0,1]$ such that
\[
	C(\alpha) = \lambda_1 v_1(\alpha) + \cdots + \lambda_m v_m(\alpha)
\]
where $\lambda_i \geq 0$ and $\sum_{i=1}^{m} \lambda_i = 1$.

In this sense, we define a \LIP{} (LIP) assignment as an expression of the form 
\[
	\Sigma = \left\{C(\alpha_i) = q_i ~|~ q_i \in [0,1], 1 \leq i \leq k\frac{}{}\right\}.
\]  
The LIP assignment is \emph{satisfiable} if there exists a convex combination $C$ on a set of valuations in the $n$-cube that jointly verifies all inequalities in $\Sigma$. This can be seen in linear algebraic terms as follows.  Given a LIP assignment $\Sigma$, let $q = (q_1, \ldots, q_k)'$ be the vector of values assigned in $\Sigma$, and suppose we are given \luka-valuations $v_1, \ldots, v_m$  and let $\lambda = (\lambda_1, \ldots, \lambda_m)'$ be a vector of $C$-coefficients.  Then consider the $k \times m$ matrix $A = [a_{ij}]$ where $a_ij = v_j(\alpha_i)$.  Then $\Sigma$ is satisfiable if there are $v_1, \ldots, v_m$ and $\lambda$ such that the set of algebraic constrains~\eqref{eq:LIPrestrictions}:

\begin{eqnarray}
	A \cdot \lambda &=& q			\nonumber \\
	\mbox{$\sum \lambda_j$} &=& 1	\label{eq:LIPrestrictions}	\\
	\lambda &\geq& 0				\nonumber
\end{eqnarray}

\noindent Conditions~\eqref{eq:LIPrestrictions} are analogous to the PSAT constraints in~\eqref{eq:PSATrestrictions}.

Note that the number $m$ of columns in $A$ is initially unknown, but the following consequence of Carathéodory's Theorem~\cite{Eck93} yields that if \eqref{eq:LIPrestrictions} has a solution, than it has a ``small'' solution.

\begin{proposition}\label{pr:LIPsmall}
	If a set of restrictions of the form \eqref{eq:LIPrestrictions} has a solution, then there are $k+1$ columns of $A$ such that the system $A_{(k+1)\times(k+1)}\lambda =q_{(k+1) \times 1}$ has a solution $\lambda \geq 0$. \qed
\end{proposition}

Given a set of pairs of formulas and bets $\tuple{\alpha_1,q_1},\ldots,\tuple{\alpha_k,q_k}$, we say that there is a \emph{\luka-Dutch book} against the bookmaker (Alice) if the gambler (Bob) can place stakes $\sigma_1, \ldots, \sigma_k \in \mathbb{Q}$ in such a way that, for all valuations $v$

\[\sum_{i=1}^{k} \sigma_i ( q_i - v(\alpha_i) ) < 0. \]

Intuitively, in a Dutch Book, Alice's bets $C(\alpha_1 ), \ldots, C(\alpha_k )$ result in financial disaster for her, for any possible world $v$.  

\begin{definition} \rm
	Given a probability assignment to propositional formulas $\{C(\alpha_i) = q_i| 1 \leq i \leq k\}$, the LIP assignment is \emph{\luka-coherent} if there are no Dutch Books against it.
\end{definition}

The following extension of de Finetti?s Dutch book theorem characterizes coherent LIP-assignments:

\begin{proposition}[Mundici~\citeyear{Mun2006}]\label{pr:lukacoherent}
	Given a LIP assignment $\Sigma=\{C(\alpha_i) = q_i| 1 \leq i \leq k\}$, the following are equivalent:
	\begin{enumerate}[(a)]
		\item $\Sigma$ is a \luka-coherent assignment.
		\item $\Sigma$ is a satisfiable LIP assignment.
	\end{enumerate}
\end{proposition}

It has been shown~\cite{BF2010} that the decision problem \luka-coherent LIP-assignments is NP-complete. So, in the case of \LIP Logic, to decide if a LIP assignment is \luka-coherent, we can again employ linear algebraic methods to solve it.  In fact, NP-completeness of LIP satisfiability can be seen as a direct corollary of Proposition~\ref{pr:LIPsmall}.  As Proposition~\ref{pr:lukacoherent} asserts that deciding \luka-coherence is the same as determining LIP assignment satisfiability, we refer to this problem as LIPSAT.

\subsection{Applying Quantitative Logic Reasoning Methods to LIPSAT}
\label{sec:applyToLIPSAT}

Based on the Quantitative Logic Reasoning approach employed in Sections~\ref{sec:pl}~and~\ref{sec:cqu}, a possible strategy to solve the LIPSAT problem is as follows.

\begin{enumerate}
	\item Generate a normal form for LIPSAT instances.
	\item Provide an algebraic formulation for a normal form LIPSAT.
	\item Develop a column generation algorithm based on the algebraic formulation.
	\item Implement the algorithm and investigate important empirical properties.
\end{enumerate}

Here we present a development of the first two items.  The last two items are currently under progress.

\subsection{Algebraic Methods for LIPSAT}
\label{sec:lipsat}

In total analogy to PSAT, define a LIPSAT instance as in \emph{(propositional) normal form} if it can be partitioned in two sets, $\tuple{\Gamma,\Psi}$, where $\Gamma = \{C(\gamma_i) = 1 | 1 \leq i \leq r\}$ and $\Psi = \{C(a_i) = q_i | a_i \textrm{ is a propositional symbol, } 1 \leq i \leq k\}$, with $0 < q_i < 1$. The partition $\Gamma$ is the satisfiable part of the normal form, usually represented only as a set of propositional formulas and $\Psi$ is the \emph{propositional LIP assignment} part. Given a LIP-assignment $\Sigma$, it is immediate that there exists a normal form LIPSAT instance $\tuple{\Gamma,\Psi}$ that is LIP-satisfiable iff $\Sigma$ is.

% Note that the format of LIP restrictions $\Sigma$ in Lemma~\ref{lem:lipsatnf} is more general than that of Proposition~\ref{pr:lukacoherent}, as the techniques for obtaining the normal form allow for such a generality.  On the other hand, in the statement of Proposition~\ref{pr:lukacoherent}, we preserve the format employed bu de Finetti. 

\begin{example}\label{ex:drunksalive}\rm
	Reconsider Example~\ref{ex:psat1} about a doctor who formulates a hypothesis on the need of at least two out of three genes $g_1, g_2, g_3$ to be active for the disease $D$ to occur.  In the classical probabilistic case, it was shown that this hypothesis was inconsistent with the fact that each gene was present in 60\% of $D$-patients. 
	
	However, if we model this problem in \LIP-logic, which allows for ``partial truths'', the hypothesis no longer contradicts the data.  In fact, we can have a formulation of the problem directly in normal form, with $\Gamma = 	\{x_1 \oplus x_2, x_1 \oplus x_3, x_2 \oplus x_3\}$ and $\Psi = \{C(x_1) = C(x_2) = C(x_3) = 0.6\}$. 
	
	This LIP assignment has many satisfying pairs of valuations and convex combination. The simplest one contains just one valuation $v_1$ such that $v_1(x_1) = v_1(x_2)= v_1(x_3) = 0.6$ and $\lambda_1 = 1$. It is immediate that $v_1$ satisfies all three formulas in $\Gamma$ and $\lambda_1 v_1$ verifies all three equalities in $\Psi$. \qed
\end{example}

The algebraic formalization of LIPSAT~\eqref{eq:LIPrestrictions} when the input LIP assignment is in normal form yields the interesting property that the columns of matrix $A$ can be extended to $\Gamma$-satisfying valuations, that is, there is a valuation $v$ over all propositional symbols in $\Gamma$ such that $v$ satisfies all formulas in $\Gamma$ and when $v$ is restricted to the symbols $a_1, \ldots, a_k$ in $\Psi$, it agrees with the respective values in $A$'s column.   

This property is used to propose a linear program that allows us to decide the LIP satisfiability of a given LIP assignment. The linear program solves ~\eqref{eq:LIPrestrictions} without explicitly representing the large matrix $A$, using once again a~\emph{column generation} method.  For that, consider the following linear program:

\begin{align} \label{eq:LIPprog}
\begin{array}{ll}
	\min 				& c' \cdot \lambda \\
	\mbox{subject to} 	& A \cdot \lambda=q \\
						& \textrm{$A$'s columns are $a_1, \ldots, a_k$ \luka-valuations } \\
						& \lambda\geq 0 \mathrm{~and~} \sum \lambda_i = 1
\end{array}
\end{align}

As in Section~\ref{sec:psatsolve}, the \emph{cost vector} $c$ in \eqref{eq:LIPprog} is a $\{0,1\}$-vector such that $c_i = 1$ iff column $A^j$ is $\Gamma$-unsatisfying.  Thus, the column generation process proceeds by generating $\Gamma$-consistent columns.  The result of this minimization  process reaches total cost $c' \cdot \pi =0$ iff the input instance is satisfiable, as stated by the following result.

\begin{theorem}\label{th:lipsatsolver}
	A normal form LIPSAT instance $\Sigma$ is LIP-satisfiable iff the corresponding linear program of the form~\eqref{eq:LIPprog} terminates with minimal total cost $c' \cdot \lambda=0$.
\end{theorem}

\begin{proof}
	$(\Leftarrow)$ If the program terminates, then clearly $\lambda$ is a convex combination of the columns of $A$ verifying the restriction in $\Sigma$.
	
	$(\Leftarrow)$ If $\Sigma$ is satisfiable, then by Proposition~\ref{pr:LIPsmall} there exists a small $k$-dimension matrix $A$ and a $\lambda$ that verifies its restrictions.  Note that $\lambda$ can be seen as a linear combination of the columns of $A$, which are \luka-valuations by \eqref{eq:LIPprog}; furthermore, $\sum \lambda_i = 1$, so $\lambda$ is a convex combination of \luka-valuations.  As column generation is able to eventually generate cost-decreasing columns, the total cost will reach 0, at which point the program terminates.
\end{proof}

\begin{corollary}[LIPSAT Complexity] \label{th:np}
	The problem of deciding the satisfiability of a LIP-assignment is NP-complete.
\end{corollary}

%\begin{proof}
%	Suppose we have a LIP-assignment of the form $\{C(\alpha_i)=1 ~|~ 1 \leq i \leq k\}$, then the problem is equivalent to deciding if the set $\{\alpha_1, \ldots, \alpha_k\}$ is \luka-satisfiable, which is NP-complete~\cite{Mun1987}.  So LIPSAT is NP-hard.
%
%	Now suppose we have a LIP-assignment, which can be placed in normal form in polynomial time by Theorem~\ref{th:nf}. Then Theorem~\ref{th:satalg} shows that if the problem is satisfiable, it can be verified in polynomial time by guessing suitable valuations and ``small distribution'' $\lambda$, constructing matrix $A_\Theta$ and verifying in polynomial time that $A_\Theta \cdot \lambda = q$.  So LIPSAT is in NP.
%\end{proof}

Despite the fact that solvable linear programs of the form~\eqref{eq:LIPprog} always have polynomial size solutions, with respect to the size of the corresponding normal form LIP-assignment, the elements of linear program itself ~\eqref{eq:LIPprog} may be exponentially large, rendering the explicit representation of matrix $A$ impractical.  

Theorem~\ref{th:lipsatsolver} serves as a basis for the development of a LIPSAT-solver and its implementation.

\subsection{A LIPSAT-solving Algorithm}
\label{sec:applyToLIPSAT}

The general strategy employed here is similar to that employed to PSAT solving~\cite{FDB2011,DF2015b}, but the column generation algorithm is considerably distinct and requires an extension of \luka-decision procedure.

From the input $\tuple{\Gamma,\Theta}$, we implicitly deal with matrix $A$ and explicitly obtain the vector of probabilities $q$ mentioned in \eqref{eq:LIPprog}. The basic idea of the simplex algorithm is to move from one feasible solution to another one with a decreasing cost. The pair $\tuple{B,\lambda}$ consisting of the basis $B$ and a LIP probability distribution $\lambda$ is a \emph{feasible solution} if $B \cdot \lambda=q$ and $\lambda\geq 0$. We assume that $q_{k+1} = 1$ such that the last line of $B$ forces $\sum_G \lambda_j = 1$, where $G$ is the set of $B$ columns that are $\Gamma$-satisfiable. Each step of the algorithm replaces one column of the feasible solution $\tuple{B^{(s-1)},\lambda^{(s-1)}}$ at step $s-1$ obtaining a new one, $\tuple{B^{(s)},\lambda^{(s)}}$. The cost vector $c^{(s)}$ is a $\{0,1\}$-vector such that $c^{(s)}_j = 1$ iff $B_j$ is $\Gamma$-unsatisfiable. The column generation and substitution is designed such that the total cost is never increasing, that is $c^{(s)}{}' \cdot \lambda^{(s)} \leq c^{(s-1)}{}' \cdot \lambda^{(s-1)}$.

Algorithm \ref{alg:LIPSATviaLSAT} presents the top level LIPSAT decision procedure. Lines \ref{line:iniini}--\ref{line:iniend} present the initialization of the algorithm. We assume the vector $q$ is in ascending order. Let the $D_{k+1}$ be a $k+1$ square matrix in which the elements on the diagonal and below are $1$ and all the others are $0$. At the initial step we make $B^{(0)} = D_{k+1}$, this forces $\lambda^{(0)}_1 = q_1 \geq 0$, $ \lambda^{(0)}_{j+1} = q_{j+1} -q_j \geq 0, 1 \leq j \leq k$; and $c^{(0)} = [c_1 \cdots c_{k+1}]'$, where $c_k=0$ if column $j$ in $B^{(0)}$ is $\Gamma$-satisfiable; otherwise $c_j=1$. Thus the initial state $s=0$ is a feasible solution.

\begin{algorithm}
	\caption{LIPSAT-CG: a LIPSAT solver via Column Generation\label{alg:LIPSATviaLSAT}}
	\textbf{Input:} A normal form LIPSAT instance
	$\tuple{\Gamma,\Theta}$.
	
	\textbf{Output:} No, if $\tuple{\Gamma,\Theta}$ is unsatisfiable.  Or a solution
	$\tuple{B,\lambda}$ that minimizes \eqref{eq:LIPprog}.
	
	\begin{algorithmic}[1]
		\STATE $q := [\{q_i ~|~ C(p_i)=q_i \in \Theta, 1 \leq i \leq k\} \cup \{1\}]$ in ascending order; \label{line:iniini}
		\STATE $B^{(0)} := D_{k+1};$ \label{lin:ini}
		\STATE $s := 0$, $\lambda^{(s)} = (B^{(0)})^{-1} \cdot q$ and $c^{(s)} = [c_1 \cdots c_{k+1}]';$ \label{line:iniend}
		\WHILE{$c^{(s)}{}' \cdot \lambda^{(s)} \neq 0$} 
		\label{lin:loop}
		\STATE $y^{(s)} = \mathit{GenerateColumn}(B^{(s)},\Gamma,c^{(s)});$ \label{lin:cond}
		\IF{$y^{(s)}$ column generation failed} 
		\RETURN No;~~ \label{lin:fail}  \COMMENT{LIPSAT instance is
			unsatisfiable}
		\ELSE
		\STATE $B^{(s+1)} = \mathit{merge}(B^{(s)}, b^{(s)})$ \label{lin:merge}
		% \STATE \COMMENT{Invariant: there is $\pi \geq 0$ with $A_{i+1}
		%   \cdot \pi = p$}  \label{lin:inv}
		%\STATE \COMMENT{Invariant: $A_{(s+1)}$ satisfies~(\ref{eq:invariant})}
		%  \label{lin:inv}
		\STATE $s\!\!+\!\!+$, recompute $\lambda^{(s)}$ and $c^{(s)};$
		\ENDIF
		\ENDWHILE\label{lin:endloop}
		\RETURN $\tuple{B^{(s)},\lambda{(s)}}$;~~  \COMMENT{LIPSAT instance is satisfiable} \label{lin:end}
	\end{algorithmic}
\end{algorithm}

Algorithm \ref{alg:LIPSATviaLSAT} main loop covers lines \ref{lin:cond}--\ref{lin:endloop} which contain the column generation strategy, detailed bellow. If column generation fails the process ends with failure in line \ref{lin:fail}. Otherwise a column is removed and the generated one is inserted in a process called \textit{merge} at line \ref{lin:merge}. The loop ends successfully when the objective function (total cost) $c^{(s)}{}' \cdot \lambda^{(s)}$ reaches zero and the algorithm outputs a probability distribution $\lambda$ and the set of $\Gamma$-satisfiable columns in $B$, at line \ref{lin:end}. 

The procedure \textit{merge} is part of the simplex method which guarantees that given a $k+1$ column $y$ and a feasible solution $\tuple{B,\lambda}$ there always exists a column $j$ in $B$ such that if  $B[j:=y]$ is obtained from $B$ by replacing column $j$ with $y$, then there is $\lambda'$ such that $\tuple{B[j:=y],\lambda'}$ is a feasible solution.

Column generation method takes as input the current basis $B$, the current cost $c$, and the \luka{} restrictions $\Gamma$; the output is a column $y$, if it exists, otherwise it signals \textbf{No}. The basic idea for column generation is the property of the simplex algorithm called the \emph{reduced cost} of inserting a column $y$ with cost $c_y$ in the basis. The reduced cost $r_y$ is given by 
\begin{align}
r_y = c_y - c'B^{-1}y
\end{align}
the objective function is non increasing if $r_y \leq 0$. The generation method always produces a column $y$ thay is $\Gamma$-satisfiable so $c_y = 0$. We thus obtain
\begin{align} \label{eq:rcost}
c'B^{-1}y \geq 0
\end{align}
which is an inequality on the elements of $y$. To force $\lambda$ to be a convex combination, we make $y_{k+1} = 1$, the remaining elements $y_i$ are valuations of the variables in $\Theta$, so that we are searching for solution to \eqref{eq:rcost} such that $0 \leq y_i \leq 1, 1 \leq i \leq k$. To finally obtain column $y$ we must extend a $\luka$-solver that generates valuations satisfying $\Gamma$ so that it also respects the linear restriction \eqref{eq:rcost}. In fact this is not an expressiveness extension of \luka{} as the McNaughton property guarantees that \eqref{eq:rcost} is equivalent to some \luka-formula on variables $y_1, \ldots, y_k$~\cite{CDM2000}. The practical details on how this can be implemented is detailed in~\cite{FP2018-prep}, which also details this final result.

\begin{theorem}
	Consider the output of Algorithm \ref{alg:LIPSATviaLSAT} with normal form input $\tuple{\Gamma,\Theta}$. If the algorithm succeeds with solution $\tuple{B,\lambda}$, then the input problem is satisfiable with distribution $\lambda$ over the valuations which are columns of $B$. If the program outputs no, then the input problem is unsatisfiable. Furthermore, there are column selection strategies that guarantee termination.
\end{theorem}

\section{Conclusion}
\label{sec:conc}

In this paper we have brought out similarities between three decisions problems in probabilistic logics and counting. The problems deal with satisfiability decision and employ similar linear algebraic methods, fine-tuned for the needs of each specific logic problem.  In this way, we belerve that we have elicited grouping them in a class which we named quantitative-logic systems.

There are several other topics were not covered in this work which pertain to all those quantitative logics dealt by this work.  Among such topics is the existence of inconsistency measurements, which have been developed for classical probabilistic theories, but not for the other systems.  Also, the existence of a phase transition for implementations of the decision procedures described here have been described, but such topic has an empirical nature and thus remains outside the scope of this article.

\subsubsection*{Acknowledgements}
This work was supported by Fapesp projects 2015/21880-4 and 2014/12236-1 and CNPq grant PQ 303609/2018-4.  This study was financed in part by the Coordena\c{c}\~{a}o de Aperfei\c{c}oamento de Pessoal de N\'{i}vel Superior -- Brasil (CAPES) -- Finance Code 001.

\bibliographystyle{chicago}
\bibliography{mf}
\end{document}